\documentclass[a4paper,UKenglish,cleveref, autoref, thm-restate]{article}

\title{Probabilistic Soft Type Assignment} 

\usepackage{enumerate}

\usepackage{stmaryrd} 
\usepackage{mathrsfs} 
\usepackage{mathpartir} 
\usepackage{bussproofs} 
\usepackage{cmll} 

\usepackage{amsfonts}
\usepackage{amssymb}
\usepackage{mathtools}
\usepackage{amsthm}
\theoremstyle{plain}
\newtheorem{remark}{Remark}
\newtheorem{example}{Example}
\newtheorem{definition}{Definition}
\newtheorem{proposition}{Proposition}
\newtheorem{theorem}{Theorem}
\newtheorem{lemma}{Lemma}
\newcommand{\lipicsEnd}{\qed}

\usepackage{framed} 
\usepackage{graphicx} 

\usepackage{longtable} 

\usepackage{xspace} 

\usepackage{hyperref} 
\usepackage{xcolor} 
\definecolor{mygreen}{rgb}{0,0.6,0}  
\definecolor{mygray}{rgb}{0.5,0.5,0.5}
\definecolor{mymauve}{rgb}{0.58,0,1}
\hypersetup{
    colorlinks,
    linkcolor={red!50!black},
    citecolor={blue!50!black},
    urlcolor={blue!80!black}  }
\usepackage{cancel}  

\usepackage[normalem]{ulem} 

\usepackage{tikz}
\usetikzlibrary{arrows.meta}
\usetikzlibrary{decorations.markings}
\usetikzlibrary{matrix}
\usetikzlibrary{arrows}
\usetikzlibrary{decorations.pathmorphing}
\usetikzlibrary{shapes}
\usetikzlibrary{arrows}
\usetikzlibrary{positioning} 
\usepackage{tikz-cd} 

\newcommand{\enquote}[1]{``{#1}''}
\newcommand{\w}[1]{\mathrm{w}({#1})}
\newcommand{\we}[2]{\mathrm{w}({#1}, {#2})}
\newcommand{\s}[1]{\vert {#1}\vert}
\newcommand{\rk}[1]{\mathrm{rk}({#1})}

\newcommand{\lamb}{\Lambda^\oc}
\newcommand{\plamb}{\Lambda^\oc_\oplus}

\newcommand{\pvalb}{\mathcal{V}}

\newcommand{\proj}{\mathtt{proj}}
\newcommand{\coin}{\mathsf{coin}}
\newcommand{\STAP}{{\textsf{PSTA}}\xspace}
\newcommand{\STA}{\textsf{STA}\xspace}
\newcommand{\SLL}{\textsf{SLL}\xspace}
\newcommand{\LAM}{\textsf{LAM}\xspace}
\newcommand{\LL}{\textsf{LL}\xspace}
\newcommand{\LLC}{\textsf{LLC}\xspace}
\newcommand{\IMALLTwo}{\textsf{IMALL}$_2$\xspace}
\newcommand{\BPP}{{\textsf{BPP}}\xspace}
\newcommand{\PP}{{\textsf{PP}}\xspace}
\newcommand{\PTM}{{\normalfont\textsf{PTM}}\xspace}
\newcommand{\pPTM}{{\normalfont\textsf{pPTM}}\xspace}
\newcommand{\ICC}{{\normalfont\textsf{ICC}}\xspace}
\newcommand{\SLR}{{\normalfont\textsf{SLR}}\xspace}

\newcommand{\FPTIME}{{\normalfont\textsf{FPTIME}}\xspace}
\newcommand{\NPTIME}{{\normalfont\textsf{NPTIME}}\xspace}

\newcommand{\PSPACE}{{\normalfont\textsf{PSPACE}}\xspace}
\newcommand{\PTIME}{{\normalfont\textsf{PTIME}}\xspace}

\begin{document}
\author{Gianluca Curzi\footnote{Corresponding author}\\
{\small\texttt Dipartmento di Informatica, Universit\`a di Torino,  Italia}\\
{\small\texttt curzi@di.unito.it}
	    \and
	    Luca Roversi\\
{\small\texttt Dipartmento di Informatica, Universit\`a di Torino,  Italia}\\
{\small\texttt roversi@di.unito.it}
}

\maketitle

\begin{abstract}
We model randomized complexity classes
in the style of Implicit Computational Complexity. We introduce \STAP, a probabilistic version of \STA, the type-theoretical counterpart of Soft Linear Logic.  \STAP is a type assignment for an extension of Simpson's Linear Lambda Calculus and its surface reduction, where Linear additives express  random choice. Linear additives are weaker than the usual ones; they allow for duplications harmlessly affecting the computational cost of  normalization. \STAP is sound and complete w.r.t. probabilistic polynomial time functions and characterizes the probabilistic complexity classes \PP and \BPP, the latter slightly less implicitly than \PP.
\end{abstract}

\section{Introduction}
\label{section:introduction}
Probabilistic complexity is a central topic in randomized computation. Many interesting  decision problems have efficient and  highly trustworthy  randomized algorithms for which no good  deterministic counterpart is known. Examples of them are in \BPP, which collects all those problems that can be solved  in polynomial time  with error probability  bounded  by a constant strictly smaller than 
$ \frac{1}{2} $. The nice point  with this class is that the error probability can be exponentially 
lowered at will  while incurring  only a polynomial slowdown, so increasing the reliability of the answer without affecting the efficiency.

We here focus on the problem of characterizing probabilistic polynomial time complexity  
classes in the style of Implicit Computational Complexity (\ICC), which 
merges arguments from computational complexity, mathematical logic and 
formal systems,  yielding machine independent characterizations of complexity classes  that do not directly  rely on  explicit bounds on the computation length. 

Starting from Mitchell et al.~\cite{mitchell1998linguistic}, several type systems were proposed to capture implicitly the probabilistic polynomial time functions by means of higher-order languages. Examples are Zhang~\cite{zhang2009computational}, or Dal Lago and Toldin~\cite{dal2015higher},  all based on Hofmann's system \SLR (Safe Linear Recursion)~\cite{hofmann1997mixed}. In particular, the latter work also discusses the inherent difficulties of characterizing the class \BPP implicitly, due to the presence of external error bounds. Recently, Seiller has proposed a promising semantic approach to \ICC based on the notion of Interaction Graphs~\cite{seiller2020probabilistic}, showing how to capture the classes $\mathsf{PL}$ (Probabilistic Logarithmic space) and $\mathsf{PP}$ (Probabilistic Polynomial time), the latter being the class of those problems that 
a probabilistic polynomial time Turing Machine solves with error probability at most 
$ \frac{1}{2} $.

Our starting observation is that all type systems introduced in~\cite{mitchell1998linguistic, zhang2009computational, dal2012probabilistic}  to characterize probabilistic polytime functions and problems share the same principles:
\begin{enumerate}[(i)]
\item   \label{enum: observation 1}  they are probabilistic higher-order generalizations of the recursion-theoretic characterization of \textsf{FPTIME} based on Bellantoni and Cook's \emph{safe recursion}~\cite{bellantoni1992new}, which  limits the 
expressive power of the recursion scheme;
\item  \label{enum: observation 2} they extend Hofmann's 
\SLR~\cite{hofmann1997mixed}, which models deterministic computations,
by means of a primitive for randomness that has the typical oracular nature.
\end{enumerate}
The goal of this paper is then twofold.  First, we model randomized computational complexity classes in the style of \ICC by exploiting those  proof-theoretical techniques derived from Girard's Linear Logic (\LL) that lead to the characterizations of   \PTIME, \FPTIME, \NPTIME and  \PSPACE~\cite{girard1998light, maurel2003nondeterministic, lafont2004soft, gaboardi2008soft}. 
A clear advantage with respect to~(\ref{enum: observation 1}) is that we deal with fully-fledged higher-order languages (and polymorphism), while in all type systems developed in~\cite{mitchell1998linguistic, zhang2009computational, dal2012probabilistic} functional arguments have to be used linearly, i.e.~at most once.

Secondly, we introduce  randomness  according to the principle that any computational step should correspond to some step of normalization in a proof system or typed calculus. Probabilistic choice is then considered as the result of an interaction between a constructor and its corresponding destructor, and so it does not depend on the answer of a \enquote{black-box}, like random primitives in~(\ref{enum: observation 2}). Matsuoka explores this idea in the non-deterministic setting~\cite{matsuoka2004nondeterministic}, introducing a self-dual additive connective into restrictions of \LL to characterize non-deterministic complexity classes. Applying the same approach  in a  probabilistic setting is less obvious,  because random choice cannot be self-dual, as recently observed by Horne in~\cite{horne2019sub}. Due to this reason, Horne proposes a Deep-inference logical system that introduces \emph{sub-additives}~\cite{horne2019sub} which enjoy De Morgan dualities and   lie ``half-way'' in between \LL additive conjunction, that models an external choice,  and  the additive disjunction, that models   the internal one. 

We achieve the above  goals by means of \STAP, a new type system that merges ideas and techniques from
Lafont's Soft Linear Logic (\SLL)~\cite{lafont2004soft}, 
Gaboardi and Ronchi Della Rocca's Soft Type Assignment (\STA)~\cite{gaboardi2009light},
Simpson's Linear Lambda Calculus (\LLC)~\cite{simpson2005reduction}, and   Ronchi Della Rocca and Roversi's calculus $\Lambda \oc $~\cite{Ronchi-Roversi:1997-STUDIA-LOGICA}. 
Probabilistic features in \STAP are expressed  by means of the interaction between a pair  $\langle M, N \rangle$ (constructor) and a new projection operator $\mathtt{proj}$ (destructor), which randomly  selects a component of $\langle M, N \rangle$. Constructor and destructor are the subject of type-assignment rules that
operate on \emph{Linear additives}, which are weaker than  standard additives. 

Linear additives trigger a restricted form of duplication that causes no exponential blow up in normalization. This way, \STAP inherits the polynomial time computational complexity bounds from \STA. Moreover, Linear additives turn out to be  expressive enough to encode the transition function of a probabilistic Turing Machine running in polynomial time, which  is the key to establish \STAP completeness with respect to the probabilistic polytime functions. The resulting characterization is fully implicit and does not depend on the choice of the reduction strategy: this is where  Linear additives play a crucial role, since the  standard additive rules require a lazy strategy  to avoid exponentially costing normalizations~\cite{girard2017proof}. 

Last, by slightly modifying the encoding of the probabilistic Turing Machine in \STAP, we can show that this system is both sound and complete w.r.t.~the complexity classes \PP and \BPP; the latter is not entirely captured implicitly due to explicit error-bounds in the statement of the characterization theorem. Perhaps a better result for \BPP is at hand by exploiting the stochastic denotational models for  deductive systems based on \LL (e.g.~probabilistic coherence spaces~\cite{danos2011probabilistic} or weighted relational semantics~\cite{laird2013weighted}), once adapting them to \STAP. 
The idea is to find a semantic characterization of \BPP in the style of~\cite{laurent2006obsessional} able to suggest some insights about the nature of this class. 
\par
Having discussed motivations about \STAP, we illustrate the key ideas behind it. 
We start from the inference rules for the additive connective $\with$  of \LL,
seen as a type-assignment:
\begin{equation}\label{eqn: additive introduction rule}
\begin{matrix}
\inferrule*[Right=$\with$I]{\Gamma \vdash M_1:A_1 \\ \Gamma \vdash M_2:A_2}{\Gamma \vdash \langle M_1, M_2 \rangle : A_1\with A_2} \quad   \quad \quad  
\inferrule*[Right=$\with$E]{\Gamma \vdash M:A_1 \with A_2 \\ i\in \{1, 2\}}{\Gamma \vdash  \pi_i(M):A_i } &
\end{matrix}
\end{equation}
The   rule $\with$I affects the complexity of normalization. Indeed,  it  gives a type to the  terms $\mathtt{add}^x_n$ defined, for all $x$ and $n \in \mathbb{N}$, as follows:
\begin{align*} 
\mathtt{add}^x_{0} \triangleq x &&
\mathtt{add}^x_{n}\triangleq 
(\lambda y. \mathtt{add}^y_{n-1})\langle x,x\rangle && 
(n>0)
\enspace .
\end{align*}
The application of $\lambda x. \mathtt{add}^x_{n}$ to some $M$ reduces to $M[n]$,  defined as:
\begin{align*}
M[0]\triangleq M&&
M[n]\triangleq \langle M[n-1], M[n-1]\rangle && 
(n>0)
\enspace .
\end{align*}
The size of $M[n]$ and the number of its redexes (if any) are  exponential with respect to those of $M$. This example shows that  linear normalization fails in presence of additive rules.

For this reason, in~\cite{Curzi2020linearaditives} the first author develops \emph{Linear additives}, weaker than standard additives, which imply a \emph{strong linear normalization property}. Linear additives come from replacing the above rule $\with$I in~\eqref{eqn: additive introduction rule} by the following one:
\begin{equation}\label{eqn: linear additive introduction rule}
\begin{matrix}
	\inferrule*[Right=$\with$I]
		{\Gamma \vdash N:A\\ x_1:A \vdash M_1: A_1 
			\\ x_2: A \vdash M_2: A_2 
			\\ \vdash U:A }
		{\Gamma \vdash 
			\mathtt{copy}^U \, N \mathtt{\ as \ }x_1,x_2 
			\mathtt{\ in \ } \langle M_1, M_2\rangle : A_1 \with A_2} 
			\end{matrix}
\end{equation}
with the proviso that $U$ is a closed and normal inhabitant of $A$, 
and the types $A, A_1, A_2$ are  free from negative occurrences of the second-order quantifier; this last proviso applies to the above $\with$E
in~\eqref{eqn: additive introduction rule} too. Intuitively, the operator $\mathtt{copy}$ \enquote{freezes} the substitutions  of $N$ in the pair 
$\langle M_1, M_2 \rangle$ until $N$ has been fully evaluated to a closed  normal form $V$. The corresponding reduction rule is then the following one:
\begin{equation}\label{eqn: copy reduction rule}
\mathtt{copy}^{U} V \mathtt{\ as \ }x_1, x_2 \mathtt{\ in \ }\langle M_1, M_2\rangle  \rightarrow \langle M_1[V/x_1], M_2[V/x_2]\rangle 
\enspace .
\end{equation}
Since the above rule duplicates normal terms only, redexes cannot be copied during reduction and linear time normalization can be recovered.  Moreover, since the type $A$ in~\eqref{eqn: linear additive introduction rule} has only finitely many closed normal inhabitants, due to the absence of $\forall$ in negative position, by always taking $U$ in~\eqref{eqn: linear additive introduction rule} as the largest term among  such inhabitants, the size of the construct $\mathtt{copy}^U$ bounds the size of the new copy of $V$; so,  normalization strictly decreases the size of terms. 

To let the reduction rule in~\eqref{eqn: copy reduction rule} preserve types, in~\cite{Curzi2020linearaditives} we introduced a further inference rule, which is $\with$I in~\eqref{eqn: additive introduction rule} with $\Gamma = \emptyset$. This rule allows  to give a type to  pairs $\langle M, N \rangle$ of closed terms. Here, for the sake of simplicity, we shall consider this rule as a special case of~\eqref{eqn: linear additive introduction rule}.  

Linear additives in \STAP justify a projection $\mathtt{proj}$, new, as compared to the standard $\pi_i$ in~\eqref{eqn: additive introduction rule}, which  non-deterministically selects a component in a pair:
\begin{equation*}
M_1 \leftarrow   \mathtt{proj}  \langle M_1, M_2 \rangle \rightarrow M_2
\enspace .
\end{equation*}

Probabilistic computation can then be expressed in \STAP by turning the one step non-deterministic reduction  $\rightarrow$ into a multi-step reduction $\Rightarrow $ between terms and probability distributions.
As expected, probabilistic choices in a higher-order calculus may lead to the failure of confluence, as distinct  evaluation strategies may produce distinct distributions. 
 
\begin{example}\label{exmp: non-confluence} 
Let $M \triangleq (\lambda x. \langle x, x\rangle) \coin$, where $\coin \triangleq  \mathsf{proj}\langle \mathbf{T}, \mathbf{F} \rangle$,  $\mathbf{T}\triangleq \lambda xy.x$ and $\mathbf{F}\triangleq \lambda xy.y$. 
A call-by-name reduction strategy first passes $\coin$ to $\lambda x. \langle x, x\rangle$. Then, it evaluates the two copies of $\coin$ produced, obtaining 
the terms 
$\langle \mathbf{T}, \mathbf{T}\rangle$, 
$\langle \mathbf{T}, \mathbf{F}\rangle$, 
$\langle \mathbf{F}, \mathbf{T}\rangle$ and $\langle \mathbf{F}, \mathbf{F}\rangle$,
as a result, each one with probability $\frac{1}{4}$. 
By contrast,  call-by-value evaluates $M$ by first reducing $\coin$, 
then passing the result to $\lambda x. \langle x, x\rangle$. 
The results are 
$\langle \mathbf{T}, \mathbf{T}\rangle$ and 
$\langle \mathbf{F}, \mathbf{F}\rangle$, both with probability $\frac{1}{2}$. 
Thus, the two parameter-passing policies give different distributions.
\lipicsEnd
\end{example}

The solution we adopt in \STAP, also studied in~\cite{diaz2018confluence, faggian2019lambda},  is to move from standard $\lambda$-calculus to Simpson's Linear Lambda Calulus (\LLC) and its surface reduction. This is an untyped term calculus $\lamb$ closely related to \LL~\cite{girard1987linear}.
It has two $\lambda$-abstractions. One is the linear abstraction $\lambda x. M$;
the other is the non-linear $\lambda \oc x. M$. The latter can duplicate arguments  with form $\oc N$, whose evaluation is suspended, according to the following rule:
\begin{equation}\label{eqn: simpson surface reduction}
 (\lambda \oc x. M)\oc N\rightarrow M[N/x]
 \enspace .
\end{equation}

Then, uniqueness of distributions in our probabilistic extension of $\lamb$ can 
be recovered. For example, $M$ in \textbf{Example~\ref{exmp: non-confluence}}
turns into $M^\oc \triangleq (\lambda \oc  x. \langle x, x\rangle) \oc \coin$. 
Since reduction is forbidden in the scope of a $\oc$ operator, $\oc \coin$ is passed 
to the function before being evaluated.
 
Unfortunately, typed variants of (extensions of) $\lamb$ may lead to the failure of Subject reduction, as the following example shows on \STA~\cite{gaboardi2009light}.

\begin{example} \label{exmp: failure subject reduction simpson}
Pretending that {\normalfont\STA} is a type-assignment for  $\lamb$, we
would have the derivation:
\begin{prooftree}
	\AxiomC{}
	\RightLabel{$ax$}
	\UnaryInfC{$x:A \vdash x:A$}
	\doubleLine
	\RightLabel{$sp$}
	\UnaryInfC{$x:\oc \oc A \vdash \oc \oc x:\oc \oc A$}
	\AxiomC{}
	\RightLabel{$ax$}
	\UnaryInfC{$y_1:A \vdash y_1:A$}
     \AxiomC{}
     \RightLabel{$ax$}
	\UnaryInfC{$y_2:A \vdash y_2:A$}
	\RightLabel{$\otimes$R}
	\BinaryInfC{$y_1: A, y_2:A \vdash \langle y_1,y_2 \rangle: A \otimes A$}
	\doubleLine
	\RightLabel{$m$}
	\UnaryInfC{$z:\oc \oc A\vdash \langle z,z \rangle: A \otimes A$}
	\RightLabel{$\multimap$I}
	\UnaryInfC{$\vdash \lambda \oc z. \langle z, z \rangle: \oc \oc A \multimap A \otimes A$}
	\RightLabel{$\multimap$E}
	\BinaryInfC{$x: \oc \oc A \vdash  (\lambda \oc z. \langle z, z \rangle) \oc \oc x:   A \otimes A$}
\end{prooftree}
where double line means multiple applications of a rule.  
Let us apply the surface reduction step in~\eqref{eqn: simpson surface reduction} to   $(\lambda \oc z. \langle z, z \rangle) \oc \oc x$.
We obtain  a judgment  $x: \oc \oc  A \vdash \langle \oc x, \oc x \rangle: A \otimes A$  without derivations in {\normalfont\STA}.  Subject reduction fails as both 
occurrences of $\oc$ in  $(\lambda \oc z. \langle z, z \rangle) \oc \oc x$ 
should be erased during  surface reduction, while only one is.
\lipicsEnd
\end{example} 
 
The last steps toward \STAP, in order to avoid the above issue, 
both introduce explicit dereliction $\mathtt{d}$, and generalize the surface 
reduction rule in~\eqref{eqn: simpson surface reduction}. 
For example, in \STAP, the conclusion of the derivation in 
\textbf{Example}~\ref{exmp: failure subject reduction simpson}  
turns  into $x: \oc \oc A \vdash(\lambda \oc z. 
\langle \mathtt{d} (\mathtt{d}(z)), \mathtt{d}(\mathtt{d}(z)) \rangle) 
\oc( \oc( \mathtt{d}(\mathtt{d}(x)))): A \otimes A$.  
Intuitively, according to the  \enquote{general} surface reduction rule,
the normalization of the term that this judgment gives a type to,
first performs a beta-reduction, yielding 
$\langle \mathtt{d} (\mathtt{d}( \oc( \oc( \mathtt{d}(\mathtt{d}(x)))))), \mathtt{d}(\mathtt{d}(\oc( \oc( \mathtt{d}(\mathtt{d}(x)))))) \rangle$;
then it rewrites each  $\mathtt{d}(\oc (M))$ into $M$. 
The resulting term is $\langle \mathtt{d}(\mathtt{d}(x)), \mathtt{d}(\mathtt{d}(x))\rangle$, with type in \STAP. 

Many proofs are postponed in the Appendix. 
\section{The type assignment system $\STAP$}
\label{The system STAoplus}

\begin{figure}[t]
	\centering
	\begin{mathpar}
		\inferrule*[Right= $ax$]
		{\\}
		{x: A \vdash x:A}
		\\
		\inferrule*[Right=$\multimap$I$l$]
		{\Gamma, x: A \vdash M: B}
		{\Gamma \vdash \lambda x.M : A \multimap  B}
		\and
		\inferrule*[Right=$\multimap$I$e$]
		{\Gamma, x: \oc \sigma \vdash M: B}
		{\Gamma \vdash \lambda \oc x.M : \oc \sigma \multimap  B}
		\and
		\inferrule*[Right= $\multimap$E]
		{\Gamma \vdash M : \sigma \multimap A
			\\ 
			\Gamma' \vdash N: \sigma}{\Gamma, \Gamma' \vdash MN: A}
		\\		
		\inferrule*[Right=$\with$I]
		{\Delta \vdash N:C\\ x_1:C \vdash M_1: C_1 
			\\ x_2: C \vdash M_2: C_2 
			\\ \vdash V:C }
		{\Delta \vdash 
			\mathtt{copy}^V \, N \mathtt{\ as \ }x_1,x_2 
			\mathtt{\ in \ } \langle M_1, M_2\rangle : C_1 \with C_2} 
			\and 
			\inferrule*[Right=$\with$E]
		{\Delta \vdash M: C\with C}
		{\Delta  \vdash \mathtt{proj}(M): C }
		\\	
		\inferrule*[Right=$sp$ ]{x_1: \sigma_1, \ldots, x_n: \sigma_n 
			\vdash M: \tau}{y_1: \oc \sigma_1, \ldots, y_n: \oc \sigma_n 
			\vdash \oc  M[\mathtt{d}(y_1)/x_1, \ldots, 
			\mathtt{d}(y_n)/x_n]:\oc \tau}  
		\and
		\inferrule*[Right=$m$]
		{\Gamma, x_1: \sigma, \ldots, x_n: \sigma \vdash M: \tau
			\ \ \  (n \geq 0)}
		{\Gamma, x: \oc\sigma \vdash M[\mathtt{d}(x)/x_1, 
			\ldots, \mathtt{d}(x)/x_n]: \ \tau}
		\\
		\inferrule*[Right=$\forall$I]
		{\Gamma \vdash M: A\langle \gamma /\alpha \rangle 
			\ \ \  \gamma \not \in \mathrm{FV}(\Gamma)}
		{\Gamma \vdash M: \forall \alpha. A}   
		\and
		\inferrule*[Right=$\forall$E]
		{\Gamma \vdash M: \forall \alpha. A}
		{\Gamma \vdash M: A\langle B/\alpha \rangle}
	\end{mathpar}
	\caption{The system $\STAP$: $C$, $C_1$, $C_2$ 
		are $\forall \oc$-lazy types, $\Delta$  is a $\forall \oc$-lazy context,
	    and $V \in \pvalb$.}
	\label{fig: the system STAplus}
\end{figure}

The type assignment system $\STAP$ is in \textbf{Figure}~\ref{fig: the system STAplus}.  It  extends $\mathsf{STA}$~\cite{gaboardi2009light} with a 
non-deterministic version of \emph{Linear additives} from~\cite{Curzi2020linearaditives} (rules  $\with$I and $\with$E).   \STAP derives judgments $ \Gamma \vdash M: \sigma $, where $\sigma$ is generated by a grammar of \emph{essential types}, like~\cite{gaboardi2009light}, $ \Gamma $ is the context that gives types to the free variables of $ M $, and $M$  belongs to the term calculus $\plamb$, which is 
Simpson's Linear Lambda Calculus (\LLC)~\cite{simpson2005reduction} 
endowed with explicit dereliction $\mathtt{d}$ (as in~\cite{Ronchi-Roversi:1997-STUDIA-LOGICA}), a  $\mathtt{copy}$ operator (as in~\cite{Curzi2020linearaditives}),  pairs  $\langle M, N \rangle$ and  a non-deterministic projection operator $\proj$.

\subsection{The types of $\STAP$}
\label{subsection:The types of STAoplus}
The following grammar generates the language of types in \STAP:
\begin{align}
\label{eqn: with grammar A}
\sigma, \tau&::= A\ \vert\ \oc\sigma
\\
\label{eqn: with grammar mathcalAe} 
A, B&::= \alpha\ \vert\ \sigma\multimap A\ \vert\ A\with\! A\ \vert\ 
\forall\alpha.A
\enspace .
\end{align}
The start symbol $ \sigma $ yields \emph{exponential types},
and $ A $ the \emph{linear types}. A type $\oc \tau$ is 
\emph{strictly exponential}. 
The set of free variables of $\sigma$ is~$FV(\sigma)$.
The meta-level substitution for types is $\sigma\langle \tau / \alpha \rangle$. 
A type $\sigma $ is \emph{closed} if $FV(A)=\emptyset$.
The \emph{$ \forall \oc$-lazy} types, crucial to prove the relevant properties 
of \STAP, are types free of negative occurrences of $\forall$ and of 
\emph{any} occurrence of $\oc$.
\begin{example} 
Typical examples of $\forall \oc$-lazy types are  the unit  $\mathbf{1}\triangleq \forall \alpha. \alpha\multimap \alpha$ and the boolean data type $\mathbf{B}\triangleq \forall \alpha. \alpha \multimap \alpha \multimap \alpha \otimes \alpha$, where  tensor $ \sigma \otimes \tau$ is introduced by means of the second-order definition $\forall \alpha. (\sigma \multimap \tau \multimap \alpha)\multimap \alpha$. Moreover, if $A$ and $B$ are $\forall \oc$-lazy types  then both $A \otimes B$ and $A \with B$ are. However, neither the type 
$\mathbf{N}\triangleq \forall \alpha. \oc (\alpha \multimap \alpha)\multimap (\alpha \multimap \alpha)$
for natural numbers, nor the type $\mathbf{B}\multimap \mathbf{B}$ are $\forall \oc$-lazy types, the former because of the occurrence of $\oc$, the latter because it has negative occurrences of $\forall$.
\lipicsEnd
\end{example}

The types in \STAP merge the structure of types from
both Soft Type Assignment (\STA)~\cite{gaboardi2009light} and
Linearly Additive Multiplicative Type Assignment (\LAM)~\cite{Curzi2020linearaditives}.
We recall that \STA is a type-assignment that 
characterizes polynomial time functions (\textsf{FPTIME}) and problems (\textsf{PTIME}) 
under the formulas-as-types paradigm. 
The types of \STA, called \emph{essential}, restrict the formulas of Soft 
Linear Logic (\SLL)~\cite{lafont2004soft} in order to assure Subject 
reduction while preserving the polynomial time bound on term normalization.
The key point about essential types is to forbid  topmost occurrences of the ``of course'' modality ``$ ! $'' in the right-hand side of an implication. I.e., $A \multimap \oc B$ is neither a type of \STA nor of \STAP. Let us also recall that \LAM~\cite{Curzi2020linearaditives} is obtained from  \textit{Intuitionistic Second-Order 
Multiplicative Additive Linear Logic} (\IMALLTwo) by replacing the standard additives with  weaker versions, called   \emph{Linear additives}, which avoid exponentially costing normalizations, typical 
of known additive rules.

\subsection{Terms and one-step surface reduction of \STAP}
\label{subsection:The terms of STAP}
The following grammar generates the language of \emph{raw terms} in \STAP:
\begin{align}
M, N&::= \mathbb{L}\ \vert\ \mathbb{A} \ \vert\ \mathtt{d}(M) \label{aligned: raw terms grammar}
\\
\mathbb{L}&::= x\ \vert\ \lambda x. M \ \vert\ \lambda \oc x. M\ \vert\ MM\ \vert\  
\oc M\  \label{aligned: linear terms grammar}
\\  
\mathbb{A}&::= \langle M, M \rangle\ \vert\ \proj(M)\ \vert\
\mathtt{copy}^V M \mathtt{\ as\ } x,y\mathtt{\ in\ } \langle M, M \rangle \label{aligned: additive terms grammar}
\\
V, U&::= x\ \vert\ \lambda x. V \  \vert\ VV \ \vert\  
\langle V, V \rangle  \label{aligned: values  grammar}
\enspace ,
\end{align}
\noindent
where $ M $ is the start symbol and
$ \mathbb{L} $ highlights the structure of terms that we take from 
\LLC. We observe that $ \mathbb{L} $ generates both a linear abstraction $\lambda x. M$ 
and a non-linear one $\lambda \oc x. M$,  the latter duplicating arguments with shape $\oc N$. 
Moreover, $ \mathbb{A} $ generates \emph{additive terms} and $ V $ gives the language in 
which we shall identify the so-called \emph{values}, as we shall see. 

The set of free variables of $M$ is $FV(M)$, where both 
$ \lambda x. M$ and $\lambda \oc x. M $ bind $x$ in $M$, and 
$\mathtt{copy}^V M \mathtt{\ as\ }x,y\mathtt{\ in\ }\langle P,Q \rangle$
binds both $x$ in $P$ and  $y$ in $ Q $. If $FV(M)=\emptyset$, then $ M $ is \emph{closed}.
The meta-level capture-avoiding substitution of $ N $ for the free variables of $ M $
  is $M[N/x]$.
The inductive definition of the \textit{size} $\s{M}$ of $M$ is standard, 
with \texttt{copy} requiring:
\begin{align}
\s{\mathtt{copy}^V M \mathtt{\ as\ }x,y\mathtt{\ in\ }
	\langle P, Q \rangle } \triangleq \s{V}+ \s{M}+ \s{P}+ \s{Q}+2 
\enspace .
\end{align}
\noindent
A variable $x$ in \emph{$M$} is \emph{surface-linear}
(\emph{$s$-linear}) if $x$ occurs free exactly once in $M$,
but not in the sub-terms $ \oc N $ and $\mathtt{d}(N)$ of $M$. A term  $M$ is \emph{surface-linear} (\emph{$s$-linear}) if both:
\begin{itemize}
	\item $x$ is $s$-linear in $N$, for every $\lambda x. N$ in $M$, and
	\item  $x$ is $s$-linear in $P$ and $y$ is $ s $-linear in $Q$,
	for every $\mathtt{copy}^{V}\, N\, \mathtt{as}\, x, y\, \mathtt{in}\, \langle P, Q\rangle$ in $M$.
\end{itemize}
\noindent
We let  $\oc ^n M$ and $\mathtt{d}^n(M)$ denote $\oc \overset{n}{\ldots}\oc M$ and~$\mathtt{d}(\overset{n}{\ldots}\mathtt{d}(M)\ldots)$, respectively.

\begin{definition}
	$\plamb$ is the language of all $s$-linear raw terms generated by the grammar~\eqref{aligned: raw terms grammar}.
\end{definition}

Since  $ \plamb $ is endowed with a  \emph{dereliction} operator $ \mathtt{d}$, that is missing in \LLC,  we need to generalize the reduction step $(\lambda \oc x. M)\oc N\rightarrow M[N/x]$ of \LLC in order to take $\mathtt{d}$ into account.

\begin{definition}[Surface-preserving substitution] 
	\label{definition:Surface-preserving substitution}
	Let $M, N \in \plamb$. 
	The \emph{surface-preserving substitution} $M \{ N/x \}$ of $N$ for the 
	free occurrences of $x$ in $M$ is:
	\begin{equation*}
	M\lbrace N/x \rbrace \triangleq \begin{cases} 
	P\lbrace Q/y \rbrace & \text{if }N=\oc Q\text{ and }  M=P[\mathtt{d}(x)/y]  \text{, with } x \not \in \mathrm{FV}(P),
	\\
	M[N/x]&\text{otherwise}
	\enspace .
	\end{cases}
	\end{equation*}
    Moreover, $M\lbrace N/x_1, \ldots, N/x_n\rbrace$  denotes $((M\lbrace N/x_1\rbrace) \ldots )\lbrace N/x_n\rbrace$.
\end{definition}

\begin{example} 
	Let us take $z\, \mathtt{d}^3(x) \,\mathtt{d}^2(x)$ in $\plamb$. The surface-preserving substitution of $\oc^2 y$ for the free occurrences 
	of $x$ in $z\, \mathtt{d}^3(x) \,\mathtt{d}^2(x)$ is:
	\begin{align*}
	& (z\, \mathtt{d}^3(x) \,\mathtt{d}^2(x))\lbrace (\oc^2 y)/x \rbrace  
	\\	
	&= (z\, \mathtt{d}^2(x') \,\mathtt{d}(x'))\lbrace (\oc y)/x' \rbrace 
	&&\text{because } z\, \mathtt{d}^3(x) \,\mathtt{d}^2(x)\triangleq  
	(z\, \mathtt{d}^2(x') \,\mathtt{d}(x'))[\mathtt{d}(x)/x']  
	\\
	&=(z\, \mathtt{d}(x'') \,x'')\lbrace y/x''\rbrace 
	&&\text{because }z\, \mathtt{d}^2(x') \,\mathtt{d}(x')\triangleq   
	(z\, \mathtt{d}(x'') \,x'')[\mathtt{d}(x')/x'']\\	&= z\, \mathtt{d}(y) \,y
	\enspace .
	&&
	\qquad\qquad\qquad\qquad\qquad
	\qquad\qquad\qquad\qquad\qquad
	\qquad\qquad\qquad
	\lipicsEnd
	\end{align*}
\end{example}

\begin{definition}
	The set $\pvalb$ of \emph{values} in $\plamb$ contains any closed term generated by the grammar~\eqref{aligned: values  grammar} that is normal with respect to the  reduction step $(\lambda x. U)V \rightarrow U[V/x]$.
\end{definition}

\begin{definition}[One-step surface reduction for $\plamb$] 
\label{defn: surface reduction for plamb} {\ }
 A \emph{surface context} is a term in $\plamb$ with a unique hole $[\cdot]$ in 
 it. The following grammar generates surface contexts:
 \begin{equation*}
 \begin{aligned}
 \mathcal{C} ::=\
 & [\cdot ]\ \vert\ 
 \lambda x.\mathcal{C}\ \vert\ 
 \lambda \oc x. \mathcal{C}\ \vert\  
 \mathcal{C}M\ \vert\ M \mathcal{C}\ \vert\  
 \mathtt{d}(\mathcal{C})\ \vert\ 
 \langle\mathcal{C}, M\rangle\ \vert\ 
 \langle M,\mathcal{C}\rangle\ \vert\ 
 \proj(\mathcal{C})\ \vert \   
 \\
 & \mathtt{copy}^V \mathcal{C} 
 \mathtt{\ as\ }x, y \mathtt{\ in\ }\langle M, N\rangle\ \vert\
 \mathtt{copy}^V M \mathtt{\ as\ }x,y\mathtt{\ in\ }
 \langle \mathcal{C}, N\rangle\ \vert\
 \\
 & 
 \mathtt{copy}^V M \mathtt{\ as\ }x,y\mathtt{\ in\ }
 \langle N, \mathcal{C}\rangle
 \enspace ,
 \end{aligned}
 \end{equation*}
 where  $\mathcal{C}[M]$ is the term obtained by filling the hole in $\mathcal{C}$ with  $ M $,
 possibly capturing free variables.

 The \emph{one-step surface reduction} 
 ${\rightarrow} \subseteq \plamb \times (\plamb)^2$ is:
 \begin{align}
 \nonumber
 (\lambda x. M)N &\rightarrow M[N/x]
 \\
 \label{eqn: forgetful reduction}
 (\lambda \oc x. M)\oc N &\rightarrow 
 M\lbrace \oc N /x\rbrace 
 \\
 \nonumber
 \proj\langle M, N \rangle &\rightarrow M, N
 \\
 \label{eqn: surface one-step duplication}
 \mathtt{copy}^{U} \, V \mathtt{\ as\ }x, y
 \mathtt{\ in\ }\langle M, N \rangle  
 &\rightarrow 
 \langle M[V/x], N[V/y]\rangle 
 && U, V \in \pvalb
 \enspace .
 \end{align}
 where $M \rightarrow N$, in fact, means $M \rightarrow N, N$, for any $ M, N $. 
 We can apply $ \rightarrow $ in surface contexts only. 
 A term of $\plamb$ is in (or is a) \textit{surface normal form}  if no 
 reduction applies to it.  Surface normal forms are ranged  over by $S$, and 
 the set of all surface normal forms  is  $\mathrm{SNF}$.
\end{definition}
\noindent
\subsection{Judgments, inference rules and derivations of \STAP}
\label{subsection:The derivation rules of STAP}
Once given the
types in Section~\ref{subsection:The types of STAoplus},
terms, values and reduction steps in Section~\ref{subsection:The terms of STAP},
comments and notations relative to the rules of $\STAP$ 
in \textbf{Figure}~\ref{fig: the system STAplus} become simpler.

Let us recall that a \emph{context} is a finite multi-set of 
\textit{assumptions} $x:A$.  If $\Gamma= x_1: A_1, \ldots, x_n:A_n $, 
then  $FV(\Gamma)\triangleq \bigcup_{i=1}^n FV(A_i)$ and  
$\vert \Gamma  \vert\triangleq \sum^n_{i=1} \vert A_i \vert$.
A context $ \Gamma $ is \emph{strictly exponential} if it contains
strictly exponential types only. A context $ \Gamma $ is
\emph{$ \forall \oc $-lazy} if it contains  $\forall \oc$-lazy 
types only. If  $\Gamma$ is $x_1:A_1, \ldots, x_n: A_n$, 
then $\oc \Gamma$ is $x_1: \oc A_1, \ldots, x_n: \oc A_n$.
By $\mathcal{D}\triangleleft \Gamma \vdash M: A$ 
we denote a derivation $\mathcal{D}$ with conclusion $\Gamma \vdash M: A$.
The \textit{size} $\vert \mathcal{D}\vert$ of a derivation $\mathcal{D}$ 
counts the number of rule instances it contains.

We conclude by commenting the inference rules of \STAP:
\begin{itemize}
\item Two introduction rules of the linear implication $\multimap$ exist.
The subject in the conclusion of $\multimap$I$l$ is $\lambda x. M$ 
and the antecedent of $\multimap$ is a linear type.
The subject in the conclusion of $\multimap$I$e$ is $\lambda \oc x. M$ 
and the antecedent of $\multimap$ is strictly exponential.
\item The \emph{linear additive rule} $\with$I replaces the standard one in~\eqref{eqn: additive introduction rule}. 
The types $C$, $C_1$, $C_2$ in
$\with$I and $\with$E must be $\forall \oc$-lazy.
Likewise, $\Delta$ is $\forall \oc$-lazy in $\with$E and $\with$I.  Finally, the term $V$ in the last premise of $\with$I is a value. 
\item  We shall consider the instance of~\eqref{eqn: additive introduction rule} with $\Gamma = \emptyset$ as a special case of $\with$I in \STAP.  This allows us to give a type to some pairs $\langle M, N \rangle$ of $\plamb$ and to let the reduction rule for $\mathtt{copy}$ in~\eqref{eqn: surface one-step duplication} preserve types in \STAP .
\item The rule $ \with $E introduces non-determinism in 
$\STAP$ by means of a projection that non-deterministically selects one of the two components in a pair.
\item Finally, $ sp $ and $ m $  come from \STA. They are the type-theoretical formulations of the logical rules  \emph{soft promotion}  and  \emph{multiplexor} of \SLL to introduce controlled duplications.
\end{itemize}

The key property of $\forall \oc$-lazy types, analogous to the one in~\cite{Curzi2020linearaditives},  
is that their size gives a bound on the size of any value that inhabits them:
\begin{proposition}
	\label{prop: values for STAoplus bounded by size type} 
	Let $V\in \pvalb$. 
	If $A$ is a $\forall \oc$-lazy type and 
	$\mathcal{D}\triangleleft  \vdash V:A$, 
	then $\vert V \vert \leq \vert A \vert$.
\end{proposition}
\begin{proof}
The statement follows by  proving  by induction on the last rule of 
	$\mathcal{D}$ the following stronger statement: ``Let $V$ be generated by~\eqref{aligned: values grammar} and normal. 
	If $x_1: A_1, \ldots, x_n: A_n \vdash M:A$,
	and $A_1\multimap \ldots \multimap A_n \multimap A$ is $\forall \oc$-lazy,
	then  $\vert M \vert \leq \sum_{i=1}^n \vert A_i \vert + \vert A \vert$''. By assumption,  the last rule of $\mathcal{D}$ cannot be $m$, $sp$, $\with$E 	or $\with$I.
\end{proof}
\noindent

\begin{remark} \label{rem: w.l.o.g. largest values}
{\normalfont\textbf{Proposition~\ref{prop: values for STAoplus bounded by size type}}} implies that, for any $\forall \oc$-lazy type $A$, a value $V$ of type $A$ exists such that $\s{U}\leq \s{V}$, for all values $U$ in the type $A$. W.l.o.g.,~we shall assume that the value $V$ in the last premise of $\with$I in~{\normalfont{Figure~\ref{fig: the system STAplus}}} has largest size among all the values of the same type. Therefore, as long as we consider typable terms in {\normalfont\STAP}, the reduction rule in~\eqref{eqn: surface one-step duplication} is such that 
$| \mathtt{copy}^{V} \, U \mathtt{\ as\ }x, y
$ $ \mathtt{\ in\ }\langle M, N \rangle| >
|\langle M[U/x], N[U/y]\rangle |$, because $V$ is a bound on the size of the new copy of the value $U$ that the reduction generates. So, Linear additives do not 
problematically affect the complexity of normalization, even though they allow duplications.
\lipicsEnd
\end{remark}
\section{A probabilistic multi-step surface reduction for \STAP}
\label{section:The probabilistic operational semantics of STAP}

We here turn the non-deterministic reduction in 
{{\textbf{Definition~\ref{defn: surface reduction for plamb}}}} into a 
probabilistic multi-step reduction relation $\Rightarrow$ 
between terms of $ \plamb $ and distributions of Surface normal forms. 

We recall that a \emph{probability distribution over a countable set $X$} 
is a function 
$f: X\to [0,1]$ such that $\sum_{x \in X}f(x)= 1$.  
The \emph{support} $\mathrm{supp}(\mathscr{D})$ of a distribution 
$\mathscr{D}$ is the subset of  all the elements in $X$ such that 
$\mathscr{D}(x)>0$.
Given $x_1, \ldots, x_n\in X$, 
then $p_1\cdot  x_1+ \ldots + p_n\cdot x_n $ denotes the distribution 
$\mathscr{D}$ with finite 
$\mathrm{supp}(\mathscr{D})=\lbrace x_1, \ldots, x_n\rbrace$, 
such that $\mathscr{D}(x_i)=p_i$, for every $i \leq n$.  Moreover,  $x\in X$ denotes both 
an element in $X$ and the distribution having all its mass on $x$, i.e.~$1 \cdot x$.
Finally, let  $ I $ be a finite set of indexes, 
let  $\lbrace p_i \rbrace_{i \in I}$ be a family of positive real numbers 
such that $\sum_{i\in I}p_i= 1$, and let 
 $\lbrace \mathscr{D}_i\rbrace_{i \in I}$ be a family of distributions.
Then, for all $x \in X$, we define 
$(\sum_{i\in I} p_i \cdot \mathscr{D}_i)(x)
\triangleq \sum_{i\in I} p_i \cdot \mathscr{D}_i(x)$.

\begin{figure}[t]
	\begin{mathpar}\mprset{vskip=0.4ex}
		\inferrule*[Right=$s1$]{S \in \mathrm{SNF}}{S \Rightarrow S} \and
		\inferrule*[Right=$s2$]{M \rightarrow M_1, M_2 \\  M_1\Rightarrow \mathscr{D}_1 \\ M_2 \Rightarrow \mathscr{D}_2 }{M \Rightarrow \textstyle{\frac{1}{2}}\cdot \mathscr{D}_1+ \textstyle{\frac{1}{2}}\cdot \mathscr{D}_2}
	\end{mathpar}
	\caption{Multi-step surface reduction $\Rightarrow$ for $\plamb$.}
	\label{fig: operational semantics for staoplus}
\end{figure}

\begin{definition}[Multi-step  surface reduction for $\plamb$]
\label{defn: small-step probabilsitic operational semantics for surface reduction} {\ }
\begin{itemize}
\item 
A \emph{surface distribution} is a probability distribution over 
$\mathrm{SNF}$ 
(see \textnormal{{\textbf{Definition~\ref{defn: surface reduction for plamb}}}}), 
i.e.~a function $\mathscr{D}: \mathrm{SNF} \longrightarrow [0,1]$ such that 
$\textstyle \sum_{S \in \mathrm{SNF}} \mathscr{D}(S)=1$.

\item 
The \emph{multi-step surface reduction $\Rightarrow$}  is the relation  between terms 
of $\plamb$   and surface distributions  defined in
{\normalfont{{Figure~\ref{fig: operational semantics for staoplus}}}}.  
Both $\pi$ and $ \rho$ range over derivations of $M \Rightarrow \mathscr{D}$.   

\item 
The  \emph{size}  $\vert \pi \vert$ of a derivation 
$\pi: M \Rightarrow \mathscr{D}$ is  $0$ if   $\pi$ is $s1$, and 
$\max(\s{\pi_1}, \s{\pi_2})+1$ if $\pi$ is $s2$ with premises  
$M \rightarrow M_1, M_2$, $\pi_1:M_1 \Rightarrow 
    \mathscr{D}_1$ and $\pi_2:M_2 \Rightarrow \mathscr{D}_2$.
\end{itemize}
\end{definition}
\vspace{.5\baselineskip}
\noindent

\begin{example} \label{exmp: Rightarrow} Consider the term $(\lambda \oc x.  \langle \coin,\mathtt{d}(x)\rangle)\oc \mathbf{I}$, where $\coin$ is as in {\normalfont{\textbf{Example~\ref{exmp: non-confluence}}}} and $\mathbf{I}\triangleq \lambda x.x$. We can apply surface reduction to this term in two different ways as in {\normalfont{{Figure~\ref{fig: different surface strategies for soundness}}}}.
In particular, 
the one with dashed lines corresponds to the derivation of the  multi-step reduction $(\lambda \oc x.  \langle \coin,\mathtt{d}(x)\rangle)\oc \mathbf{I}   \Rightarrow  \frac{1}{2}\cdot  \langle \mathbf{T}, \mathbf{I} \rangle + \frac{1}{2}\cdot  \langle \mathbf{F}, \mathbf{I} \rangle$ in 
{\normalfont{{Figure~\ref{fig: example of multi-step derivation}}}}.
\lipicsEnd
\end{example}

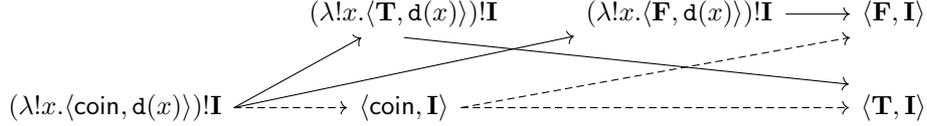
\begin{figure}[t]
	\begin{tikzcd}[ampersand replacement = \&]
		\&(\lambda \oc x.  \langle \mathbf{T},\mathtt{d}(x)\rangle)\oc \mathbf{I} \arrow[rrd,start anchor=south, end anchor=north west] \& (\lambda \oc x.  \langle \mathbf{F},\mathtt{d}(x)\rangle)\oc \mathbf{I} \arrow[r]\&  \langle \mathbf{F},\mathbf I \rangle \\
		(\lambda \oc x.  \langle \coin,\mathtt{d}(x)\rangle)\oc \mathbf{I} \arrow[r, dashed] \arrow[ru, start anchor=east]\arrow[rru, start anchor=east]\&   \langle \coin,\mathbf I \rangle \arrow[rr, dashed]\arrow[rru,start anchor=east, end anchor = south west, dashed]\&\& \langle \mathbf{T},\mathbf I \rangle 
	\end{tikzcd}
	\caption{Different surface reduction strategies for $ (\lambda \oc x.  \langle \coin,\mathtt{d}(x)\rangle)\oc \mathbf{I}$,  where $\coin \rightarrow \mathbf{T}, \mathbf{F}$.}
	\label{fig: different surface strategies for soundness}
\end{figure}

\begin{figure}[t]
	\def\defaultHypSeparation{\hskip .17cm}
	\def\ScoreOverhang{0cm}
	\scalebox{.9}{
	\AxiomC{$(\lambda \oc x.  \langle \coin,\mathtt{d}(x)\rangle)\oc \mathbf{I}   \rightarrow   \langle \coin, \mathbf{I}  \rangle  $}
	\AxiomC{$  \langle \coin, \mathbf{I} \rangle \rightarrow    \langle \mathbf{T}, \mathbf{I} \rangle, \,  \langle \mathbf{F}, \mathbf{I} \rangle$}
	\AxiomC{}
	\UnaryInfC{$  \langle \mathbf{T}, \mathbf{I} \rangle \Rightarrow   \langle \mathbf{T}, \mathbf{I} \rangle$}
	\AxiomC{}
	\UnaryInfC{$  \langle \mathbf{F}, \mathbf{I} \rangle \Rightarrow   \langle \mathbf{F}, \mathbf{I} \rangle$}
	\TrinaryInfC{$ \langle \coin, \mathbf{I} \rangle\Rightarrow  \frac{1}{2}\cdot  \langle \mathbf{T}, \mathbf{I} \rangle + \frac{1}{2}\cdot  \langle \mathbf{F}, \mathbf{I} \rangle$}
	\BinaryInfC{$(\lambda \oc x.  \langle \coin,\mathtt{d}(x)\rangle)\oc \mathbf{I}   \Rightarrow  \frac{1}{2}\cdot  \langle \mathbf{T}, \mathbf{I} \rangle + \frac{1}{2}\cdot  \langle \mathbf{F}, \mathbf{I} \rangle$}
	\DisplayProof
		}
\caption{Derivation of $(\lambda \oc x.  \langle \coin,\mathtt{d}(x)\rangle)\oc \mathbf{I}   \Rightarrow  \frac{1}{2}\cdot  \langle \mathbf{T}, \mathbf{I} \rangle + \frac{1}{2}\cdot  \langle \mathbf{F}, \mathbf{I} \rangle$, where $\coin \rightarrow \mathbf{T}, \mathbf{F}$.}
	\label{fig: example of multi-step derivation}
\end{figure}

\begin{example}
\label{exmp: Rightarrow may be undefined} 
Let $\mathbf{\Omega_\oc}\triangleq \mathbf{\Delta}_\oc (\oc \mathbf{\Delta}_\oc)$, where  $\mathbf{\Delta}_\oc\triangleq \lambda \oc x. \mathtt{d}(x) \oc \mathtt{d}(x)$. Since  
$\mathbf{\Omega_\oc} \rightarrow 
( \mathtt{d}(x) \oc \mathtt{d}(x))\lbrace\oc\mathbf{\Delta}_\oc/x \rbrace$ $ 
  = (( y \, \oc y)[\mathtt{d}(x)/y])\lbrace \oc\mathbf{\Delta}_\oc/x \rbrace  
  = \mathbf{\Omega_\oc}$, no surface distribution $\mathscr{D}$ exists such that $\mathbf{\Omega_\oc}\Rightarrow \mathscr{D}$.
\lipicsEnd
\end{example}

The calculus $\plamb$ enjoys the following confluence property:
\begin{theorem}[Confluence for $\plamb$]
	\label{thm: confluence for lamb}
	Let $M \in \plamb$. If $M \Rightarrow \mathscr{D}$ and 
	$M \Rightarrow \mathscr{E}$ then $\mathscr{D}= \mathscr{E}$.
\end{theorem}
\begin{proof}[sketch]
	Following~\cite{dal2015higher}, we define a relation $\Rrightarrow$ between 
	terms and distributions over $\plamb$, where rule $s1$ is relaxed  to allow  
	$M \Rrightarrow M$, for all $M \in \plamb$, and such that
	${\Rightarrow}\subseteq {\Rrightarrow}$.
	So, if $\Rrightarrow$ is confluent, then 
	$\Rightarrow$ is. To show this, we first establish confluence for  $\rightarrow$,  which requires to prove 
	\enquote{If $M \rightarrow M'_1, M'_2$ and $M \rightarrow M''_1, M''_2$,  
	then there exist $N_1, N_2, N_3, N_4$ distinct such that 
	$M'_1 \rightarrow N_1, N_2$, $M'_2 \rightarrow N_3, N_4$, and 
	$\exists i\in \lbrace 1,2\rbrace$ such that $M''_i \rightarrow N_1, N_3$, and 
	$M''_{3-i}\rightarrow N_2, N_4$} among other lemmas. 
    Then, we lift this confluence property  from	$\rightarrow$ to $\Rrightarrow$. 
\end{proof}

\section{Probabilistic Polytime Soundness of \STAP}
\label{sec4}
We show that the evaluation of any term of $ \plamb $ with type in \STAP (according to the multi-step reduction $\Rightarrow$) can be simulated by a \emph{polynomial time} Probabilistic  Turing Machine (\pPTM), i.e. by a Probabilistic Turing Machine (\PTM) whose running time  is bounded by some   polynomial in the input size.
We adapt the proof developed for  $\mathsf{STA}$~\cite{gaboardi2009light}, known since~\cite{lafont2004soft}, to the probabilistic setting.  
We show that Surface reduction preserves types and 
shrinks the weight of derivations; so, in fact, we prove a version of Subject 
reduction 
({\textbf{Theorem~\ref{thm: weighted subject reduction}}}) 
a bit stronger than usual.
From this we derive that the number of surface reduction steps  rewriting a 
typable term into its surface normal forms 
is polynomially bounded 
({\textbf{Lemma~\ref{lem: polystep soundness complexity}}}.)
This, eventually, implies Probabilistic Polytime Soundness 
({\textbf{Theorem~\ref{thm: polytime soundness complexity}}}.)

We start recalling the notions of rank (here $m$-rank) and depth from~\cite{lafont2004soft, gaboardi2009light}. We 
introduce the $sp$-rank; the treatment of both
$\oc$ and $\mathtt{d}$, which affect the size of a term,
requires it.

\begin{definition}[$m$-rank, $sp$-rank, depth]{\ }
\begin{itemize}
\item 
The \emph{$m$-rank} of a rule $m$ of the form:
\begin{prooftree}
\AxiomC{$\Gamma, x_1: \sigma, \ldots, x_n: \sigma \vdash M: \tau$}
\AxiomC{$(n \geq 0)$}
\RightLabel{$m$}
\BinaryInfC{$\Gamma, x: \oc\sigma \vdash M[\mathtt{d}(x)/x_1, \ldots, 
	\mathtt{d}(x)/x_n]: \tau$}
\end{prooftree}
is the number $k \leq n$ of variables $x_i$ such that $x_i \in FV(M)$. 
The \emph{$m$-rank} $\mathrm{rk}(\mathcal{D})$ of a derivation $\mathcal{D}$ is 
$\max(1, k)$, with $k$ the maximum $m$-rank among the instances  of 
$m$ in $\mathcal{D}$.

\item 
The \emph{$sp$-rank} of a rule $sp$ of the form: 
\begin{prooftree}
\AxiomC{$ x_1: \sigma_1, \ldots, x_n: \sigma_n \vdash M: \tau$}
\RightLabel{$sp$}
\UnaryInfC{$ y_1: \oc \sigma_1, \ldots, y_n: \oc \sigma_n \vdash \oc M[\mathtt{d}(y_1)/x_1, \ldots, \mathtt{d}(y_n)/y_n]: \tau$}
\end{prooftree}
is the number $k \leq n$ of variables $x_i$ such that $x_i \in FV(M)$. 
\item 
The \emph{depth} $\mathrm{d}(\mathcal{D})$ of a derivation $\mathcal{D}$ is 
the maximum number of occurrences of $sp$ in a path from the 
conclusion of $ \mathcal{D} $ to one axiom in $\mathcal{D}$. 
\end{itemize}
\end{definition}

\begin{definition}[Weight]
Let $r \geq 1$. 
The \emph{weight} $\mathrm{w}(\mathcal{D}, r)$ (relative to $ r $)
of a derivation $\mathcal{D}$ is defined by structural induction 
on $\mathcal{D}$:
\begin{itemize}
\item if the last rule of $\mathcal{D}$ is $ax$, then $\mathrm{w}(\mathcal{D}, r)=1$;
\item if $\mathcal{D}$ is obtained from $\mathcal{D}'$ by applying ${\multimap}\mathrm{I}$l, ${\multimap}\mathrm{I}$e or $\with \mathrm{E}$, then $\mathrm{w}(\mathcal{D}, r)= \mathrm{w}(\mathcal{D}', r)+1$;
\item if $\mathcal{D}$ is obtained from $\mathcal{D}'$ and $\mathcal{D}''$ by applying ${\multimap}\mathrm{E}$,
then $\mathrm{w}(\mathcal{D}, r)
          = \mathrm{w}(\mathcal{D}', r)+ \mathrm{w}(\mathcal{D}'', r)+1$;
\item if $\mathcal{D}$ is obtained from $\mathcal{D}_1$, $\mathcal{D}_2$, $\mathcal{D}_3$, and $\mathcal{D}_4$ by applying $\with\mathrm{I}$, then  
 $\mathrm{w}(\mathcal{D}, r)=\mathrm{w}(\mathcal{D}_1,r)
+\mathrm{w}(\mathcal{D}_2, r)
+\mathrm{w}(\mathcal{D}_3, r)
+\mathrm{w}(\mathcal{D}_4, r)
+2 $;
\item if $\mathcal{D}$ is obtained from $\mathcal{D}'$ by applying $\forall\mathrm{I}$, or $\forall\mathrm{E}$,  
then $\mathrm{w}(\mathcal{D},r)= \mathrm{w}(\mathcal{D}', r)$;
\item if $\mathcal{D}$ is obtained from $\mathcal{D}'$ by applying $sp$ with  $sp$-rank $k$, 
then $\mathrm{w}(\mathcal{D}, r)= r\cdot (\mathrm{w}(\mathcal{D}', r)+k)+1$;
\item  if $\mathcal{D}$ is obtained from $\mathcal{D}'$ by applying   $m$ with $m$-rank $k$,
then $\mathrm{w}(\mathcal{D}, r)=  \mathrm{w}(\mathcal{D}', r)+k$. 
\end{itemize}
\end{definition}

\begin{lemma}[\cite{gaboardi2009light}]
\label{lem: properties of weight} 
Let $r \geq 1$ and $\mathcal{D}\triangleleft \Gamma 
\vdash M: \sigma$. 
Then:
\begin{enumerate}
\item 
\label{enum: properties of weight 1} 
$\mathrm{rk}(\mathcal{D})\leq \vert M \vert$;
\item 
\label{enum: properties of weight 2} 
$\mathrm{w}(\mathcal{D}, r)\leq r^{\mathrm{d}(\mathcal{D})}\cdot \mathrm{w}(\mathcal{D}, 1)$;
\item 
\label{enum: properties of weight 3} 
$\mathrm{w}(\mathcal{D}, 1)= \vert M \vert$. 
Moreover, if $\mathcal{D}$  has no occurrences of $sp$ and $m$, 
then $\w{\mathcal{D}, r}= \vert M \vert$.
\end{enumerate}
\end{lemma}

\begin{theorem}[Weighted Subject reduction]
\label{thm: weighted subject reduction} 
Let  $\mathcal{D}\triangleleft \Gamma \vdash M: \sigma$ and 
$r \geq \mathrm{rk}(\mathcal{D})$. 
If $M \rightarrow M_1, M_2$, then there exist $\mathcal{D}_1$ and $\mathcal{D}_2$ such that:
\begin{enumerate}
\item $\mathcal{D}_i \triangleleft \Gamma \vdash M_i: \sigma$.
\item $\mathrm{w}(\mathcal{D}_i, r)<\mathrm{w}(\mathcal{D}, r)$, for $i \in \lbrace 1,2 \rbrace$.
\end{enumerate}
\end{theorem}
\begin{proof}[sketch] 
The proof is by induction on the definition of the one-step reduction relation $\rightarrow$. 
It requires to prove a \emph{Weighted Substitution property}: 
``For all $r \geq \mathrm{rk}(\mathcal{D}_1)$, 
if  $\mathcal{D}_1 \triangleleft \Gamma, x: \sigma \vdash M: \tau$ 
and $\mathcal{D}_2 \triangleleft \Delta \vdash N: \sigma$ then  
$\mathcal{D}^*$ exists such that both
$\mathcal{D}^* \triangleleft \Gamma, \Delta \vdash M\lbrace N/x\rbrace : \tau$  
and $\mathrm{w}(\mathcal{D}^*, r )\leq \mathrm{w}(\mathcal{D}_1,r)+ \mathrm{w}(\mathcal{D}_2, r)$.'' 
The proof of the Weighted Substitution property relies on 
the lemma: 
``If $\Gamma \vdash M: \oc \sigma$ is derivable in $\STAP$, 
then  $\Gamma$ is a strictly exponential context.''  
\end{proof}

\par
The above theorem implies that terms typable in \STAP are strong normalizing with respect to Surface reduction $\rightarrow$, and hence that, for any $M$ with type in \STAP,  a surface distribution $\mathscr{D}$ exists such that $M \Rightarrow \mathscr{D}$. By {{\textbf{Theorem~\ref{thm: confluence for lamb}}}}, this surface distribution is \emph{unique}.

Every derivation $ M \Rightarrow \mathcal{D} $, with $ M $ having a type in 
\STAP, enjoys the following:

\begin{lemma}[Uniformity]
\label{lem: all derivations of Rightarrow have same size}
Let $\Gamma \vdash M: \sigma$. 
If $\pi':M\Rightarrow \mathscr{D}$ and $ \pi'':M\Rightarrow \mathscr{D}$, 
then $\s{\pi'}=\s{\pi''}$. 
\end{lemma}
\begin{proof}[sketch]
Reductions take place at a \enquote{surface level}, i.e.~never in the scope of any $\oc$, so  that redexes are never duplicated or erased.
\end{proof}
\noindent
The above lemma says that an upper bound on $M \Rightarrow \mathscr{D}$
exists on the length of each non-deterministic branching of all possible 
reduction strategies applied to $M$. That bound is limited by a polynomial 
in the size of $M$:

\begin{lemma}[Strong polystep soundness]
\label{lem: polystep soundness complexity} 
Let $\mathcal{D}\triangleleft \Gamma \vdash M: \sigma$  and $\pi: M \Rightarrow \mathscr{D}$. Then:
\begin{enumerate}
\item 
\label{enum: soundness complexity 1} 
$\vert \pi \vert  \leq \vert M \vert^{\mathrm{d}(\mathcal{D})+1}$.
\item 
\label{enum: soundness complexity 2} 
$\vert N \vert \leq \vert M \vert^{\mathrm{d}(\mathcal{D})+1}$, 
for every $N \rightarrow N', N''$ premise of $s2$  in $\pi$.
\end{enumerate}
\end{lemma}
\begin{proof}
Let $\mathcal{D}\triangleleft \Gamma \vdash M: \sigma$. 
{\textbf{Lemma~\ref{lem: properties of weight}.\ref{enum: properties of weight 1}-\ref{enum: properties of weight 3}}} implies:
\allowdisplaybreaks
\begin{equation*}
 \mathrm{w}(\mathcal{D}, \mathrm{rk}(\mathcal{D}))\leq \mathrm{w}(\mathcal{D}, \vert M \vert) \leq \vert M \vert^{\mathrm{d}(\mathcal{D})} \cdot \mathrm{w}(\mathcal{D}, 1) = \vert M \vert^{\mathrm{d}(\mathcal{D})}\cdot \vert M \vert = \vert M \vert ^{\mathrm{d}(\mathcal{D})+1} .
\end{equation*}
By induction on the size of $ \pi: M \Rightarrow \mathscr{D } $,
for all  $r \geq \rk{\mathcal{D}}$, we can prove:
\begin{enumerate}[i.]
\item 
\label{enum: simplify 1} 
$\vert \pi \vert \leq \mathrm{w}(\mathcal{D},r)$;
\item  
\label{enum: simplify 2}
$\vert N\vert \leq \mathrm{w}(\mathcal{D}, r)$, 
for every  $N \rightarrow N', N''$ 
premise of $s2$ in $\pi$. 
\end{enumerate}
If the last rule of $\pi$ is $s1$, then 
both~{{\ref{enum: simplify 1}}} 
and~{{\ref{enum: simplify 2}}} here above hold trivially. 
Otherwise, the last rule of $ \pi $ is $s2$ with premises 
$M \rightarrow M_1, M_2$, 
$\pi_1:M_1 \Rightarrow \mathscr{D}_1$, 
and $\pi_2:M_2 \Rightarrow \mathscr{D}_2$.
By {\textbf{Theorem~\ref{thm: weighted subject reduction}}}, there exist 
 $\mathcal{D}_1$ and $\mathcal{D}_2$ 
such that both $\mathcal{D}_i \triangleleft \Gamma \vdash M_i : \sigma$ 
and $\mathrm{w}(\mathcal{D}_i, r)< \mathrm{w}(\mathcal{D}, r)$.  
Concerning point~{{\ref{enum: simplify 1}}}, 
by induction, 
$\vert \pi_i  \vert \leq \mathrm{w}(\mathcal{D}_i,r)$, 
with $i\in \lbrace 1, 2 \rbrace$.  
Hence,
$\vert \pi \vert= \max (\vert \pi_1 \vert, \vert \pi_2 \vert)+1 \leq  
 \max (\mathrm{w}(\mathcal{D}_1, r), 
       \mathrm{w}(\mathcal{D}_2,r))+1\leq  
 \mathrm{w}(\mathcal{D},r)$.
Concerning point~{{\ref{enum: simplify 2}}}, 
$\vert N \vert \leq \mathrm{w}(\mathcal{D}_i, r)<\mathrm{w}(\mathcal{D}, r)$
holds by induction, for all $i \in \lbrace 1, 2\rbrace$ and 
for all $N \rightarrow N'_i, N''_i$, premise of some $s2$ in $\pi_i$. 
Finally, 
by {\textbf{Lemma~\ref{lem: properties of weight}.\ref{enum: properties of weight 3}}}, we have 
$\vert M \vert = \mathrm{w}(\mathcal{D}, 1)\leq \mathrm{w}(\mathcal{D}, r)$. 
\end{proof} 

\begin{remark} 
\label{rem: reduction rules simulated PTM}  
From~\cite{terui2001light}, we know that a Turing Machine simulates a 
$\beta$-reduction $M \rightarrow_\beta M'$ in a time bounded by
$\mathcal{O}(\vert M \vert^2)$. 
Similarly, for every step $M \rightarrow M_1, M_2$ in 
{\normalfont\textbf{Definition~\ref{defn: surface reduction for plamb}}},  
a {\PTM} exists which, receiving an encoding of
$M$ as input, produces an encoding of $M_i$ as output with 
probability a half, in a time bounded by $\mathcal{O}(\vert M \vert^2)$.
\lipicsEnd
\end{remark} 

\begin{theorem}[Probabilistic Polytime Soundness of \STAP] 
\label{thm: polytime soundness complexity}
Let $\mathcal{D}\triangleleft \Gamma \vdash M: \sigma$ 
be such that $M \Rightarrow \mathscr{D}$. 
A {\normalfont\PTM} $ \mathcal{P} $ exists such that,
for all $S \in \text{supp}(\mathscr{D})$:
\begin{itemize}
\item $ \mathcal{P} $ takes an encoding of $M$ as input
and produces an encoding of the surface normal form $S$ 
as output, with probability $\mathscr{D}(S)$, and

\item $ \mathcal{P} $ runs in a time bounded by
$\mathcal{O}(\vert M \vert^{3( \mathrm{d}(\mathcal{D})+1)})$,
i.e. $ \mathcal{P} $ is a {\normalfont\pPTM}.
\end{itemize}  
\end{theorem}
\begin{proof}
By {\textbf{Lemma~\ref{lem: polystep soundness complexity}.\ref{enum: soundness complexity 2}}} and  
\textsf{Remark~\ref{rem: reduction rules simulated PTM}}, 
each reduction step $P\rightarrow P_1, P_2$, premise of $s2$ in $\pi:M \Rightarrow \mathscr{D}$, 
can be simulated by a \PTM that runs in a time
bounded by $\mathcal{O}(\vert M \vert ^{2(\mathrm{d}(\mathcal{D})+1)})$.  
By 
{\textbf{Lemma~\ref{lem: polystep soundness complexity}.\ref{enum: soundness complexity 1}}}
there can be at most $\mathcal{O}(\vert M \vert^{\mathrm{d}(\mathcal{D})+1})$ 
instances of $s2$ in $\pi$. 
So, a {\PTM} exists that simulates the evaluation of 
$M$, running in time bounded by
$\mathcal{O}(\vert M \vert^{3( \mathrm{d}(\mathcal{D})+1)})$.
\end{proof}
\section{Probabilistic Polytime Completeness of \STAP}
\label{section:Probabilistic Polytime Completeness}
We prove that the terms of $ \plamb $ with a type in \STAP
are expressive enough to encode any \emph{polynomial time} Probabilistic 
Turing Machine (\pPTM), i.e. a Probabilistic Turing Machine (\PTM) whose running time  is bounded by some   polynomial in the input size. This allows us to show that \STAP is complete with respect to the functions computed  by the \pPTM.
Typically, encoding a Turing Machine by means of ($ \lambda$-)terms
requires to represent configurations, 
transitions between configurations, 
a phase of initialization, and one of output extraction. 
Here we focus on the main details of the key step to get completeness,
i.e. the definition of the transition function of
any {\pPTM} in \STAP. 
 
To that purpose, we recall that tensors 
($\otimes$) and  unit ($\mathbf{1}$) exist 
in \STAP as second-order types (see~\cite{mairson2003computational}
for example.) So, inference rules for $\otimes$ and $\mathbf{1}$ are 
derivable and we can fairly assume that the reduction rules
$\mathtt{let}\ \mathbf{I}\mathtt{\ be\ }\mathbf{I}
 \mathtt{\ in\ }N \rightarrow_\beta N$ and 
$\mathtt{let}\ M_1\otimes  M_2\mathtt{\ be\ }x_1\otimes x_2 
 \mathtt{\ in\ }N\rightarrow_\beta N[M_1/x_1, M_2/x_2]$
are available. Given tensors and unit,
the types and terms of \STAP:
 \begin{equation}\label{eqn: booleans}
 \mathbf{B}\triangleq \forall \alpha.( \alpha\multimap \alpha \multimap \alpha \otimes \alpha ) \qquad \qquad \underline{0}\triangleq \lambda xy. x\otimes y \qquad  \qquad \underline{1}\triangleq \lambda xy. y \otimes x
 \end{equation}
can represent booleans~\cite{mairson2003computational}. 
As a notation, $\mathbf{B}^n$ stands for 
$\mathbf{B}\otimes \overset{n}{\ldots}\otimes \mathbf{B}$ and 
$\mathtt{0}^n$ (resp.~$\mathtt{1}^n$) 
for $\mathtt{0}\otimes \overset{n}{\ldots}\otimes \mathtt{0}$ 
(resp.~$\mathtt{1}\otimes \overset{n}{\ldots}\otimes \mathtt{1}$.)
 
\begin{figure}[t]
	\def\ScoreOverhang{0cm}
	\def\defaultHypSeparation{\hskip .2cm}
	\AxiomC{}
	\RightLabel{$ax$}
	\UnaryInfC{$x: \mathbf{B}^{n+1} \vdash x: \mathbf{B}^{n+1} $}
	\AxiomC{\vdots \phantom{\quad ($i =0,1)$}}
	\noLine
	\UnaryInfC{$x_i:\mathbf{B}^{n+1}  \vdash \underline{\delta_i}\, x_i:  \mathbf{B}^{n+2}  \quad (i =0,1)$}
	\AxiomC{\vdots}
	\noLine
	\UnaryInfC{$\vdash \mathtt{0}^{n+1}  : \mathbf{B}^{n+1}$}
	\RightLabel{$\with$I}
	\TrinaryInfC{$x: \mathbf{B}^{n+1} \vdash\mathtt{copy}^{\mathtt{0}^{n+1}} x \mathtt{ \ as \ }x_0, x_1 \mathtt{ \ in  \ }\langle \underline{\delta_0}\,x_0, \underline{\delta_1}\,x_1\rangle :  \mathbf{B}^{n+2} \with  \mathbf{B}^{n+2} $}
	\RightLabel{$\with$E}
	\UnaryInfC{$x: \mathbf{B}^{n+1} \vdash \mathtt{proj} \, (\mathtt{copy}^{\mathtt{0}^{n+1}} x \mathtt{ \ as \ }x_0, x_1 \mathtt{ \ in  \ }\langle \underline{\delta_0}\,x_0, \underline{\delta_1}\,x_1\rangle ):  \mathbf{B}^{n+2} $}
	\RightLabel{$\multimap$I$l$}
	\UnaryInfC{$\vdash \lambda x.\mathtt{proj} \, (\mathtt{copy}^{\mathtt{0}^{n+1}} x \mathtt{ \ as \ }x_0, x_1 \mathtt{ \ in  \ }\langle \underline{\delta_0}\,x_0, \underline{\delta_1}\,x_1\rangle ): \mathbf{B}^{n+1} \multimap \mathbf{B}^{n+2} $}
	\DisplayProof
	\caption{Derivation of $\underline{\delta_{\mathcal{P}}}$ in $\STAP$}
	\label{fig: derivation of probabilistic transition function}
\end{figure}

We recall that the transition function $\delta_\mathcal{P}$ of
a {\PTM} $\mathcal{P}$ can be seen as superposing the transition functions 
$\delta_0$ and $\delta_1$ of two \emph{deterministic} Turing Machines;
every computation step of $\mathcal{P}$ selects one between 
$\delta_0$ and $\delta_1$ with probability $\frac{1}{2}$. 
So, let $\delta_0, \delta_1: Q \times \lbrace0, 1 \rbrace 
\longrightarrow Q \times \lbrace0, 1 \rbrace \times \lbrace \text{left}, 
\text{right} \rbrace$ be the transition functions of two deterministic 
Turing Machines $ \mathcal{M}_1 $ and $ \mathcal{M}_2 $
with $Q$ containing at most $2^n$ states.
Following~\cite{gaboardi2009light}, these transition functions can be encoded by  suitable terms  $\underline{\delta_1}$ and $\underline{\delta_2}$  of type $\mathbf{B}^{n+1} \multimap \mathbf{B}^{n+2}$. We can define:
\begin{equation}
\label{equation: transition function PTM}
\underline{\delta_{\mathcal{P}}}
\triangleq 
\lambda x.\mathtt{proj}\,
(\mathtt{copy}^{\mathtt{0}^{n+1}} x \mathtt{\ as\ } x_0, x_1 
 \mathtt{\ in\ }\langle 
 \underline{\delta_0}\,x_0, \underline{\delta_1}\,x_1\rangle ) 
\enspace ,
\end{equation}
the transition function of $\mathcal{P}$, whose derivation in \STAP is in 
\textsf{\textbf{Figure~\ref{fig: derivation of probabilistic transition function}}}. Let $ \underline{q}\otimes \underline{p}  $ be a pair that encodes 
	the configuration  $(q, b) \in Q \times \{ 0,1\}$ 
	of a \PTM.
	Let $\underline{\delta_i}  \, \underline{q}\otimes \underline{p}  
	\Rightarrow   \underline{q_i}\otimes \underline{b_i}\otimes \underline{m_i}$ for $i \in \lbrace 1,2\rbrace $.
	Then $\underline{\delta_{\mathcal{P}}}\,  \underline{q}\otimes \underline{p}  \Rightarrow \textstyle{ \frac{1}{2}}\cdot  \underline{q_0}\otimes \underline{b_0}\otimes \underline{m_0} +\textstyle{\frac{1}{2}}\cdot  \underline{q_1}\otimes \underline{b_1}\otimes  \underline{m_1}  $ can be easily derived.

Having an encoding for  booleans, we can now show how to represent boolean strings in \STAP. For every $i \geq 1$,
the \emph{indexed type} $\mathbf{S}_i$ and the 
\emph{indexed $n$-ary boolean strings} $ \underline{s}_i$,
whose type is $\mathbf{S}_i$, exist in \STAP:
\begin{equation}\label{eqn: indexed strings}
\begin{aligned}
\mathbf{S}_i &\triangleq \forall \alpha. \oc ^i (\mathbf{B} \multimap \alpha \multimap \alpha) \multimap (\alpha \multimap \alpha)\\
\underline{s}_i &\triangleq \lambda \oc c. \lambda z. 
\mathtt{d}^i(c)\,\underline{b_1}(\ldots(\mathtt{d}^i(c)\,\underline{b_n} z ) \ldots)  
&& \text{where }  s=b_1 \ldots b_n \in 
\lbrace 0, 1 \rbrace^n \text{ and } n \in \mathbb{N}
\enspace .
\end{aligned}
\end{equation}  
If $n=1$, we write $\mathbf{S}$ (resp.~$  \underline{s}$) in place of 
$\mathbf{S}_1$ (resp.~$ \underline{s}_1$).
The need to introduce families of terms and families of types in~\eqref{eqn: indexed strings}  is due to the inference rule $m$, as already noticed in~\cite{gaboardi2009light}.

The following result states that \STAP characterizes the functions computed in polynomial time by a \PTM.

\begin{theorem}[Probabilistic Polytime Completeness of \STAP]
\label{thm: completeness complexity PTM} 
Let $\mathcal{P}$ be a \PTM. If:
\begin{enumerate}
\item \label{thm: completeness complexity PTM 1}
$\mathcal{P}$ runs in $p(n)$-time, for some polynomial $p: \mathbb{N}\to \mathbb{N}$ with  $\deg(p)=d_1$, and
\item \label{thm: completeness complexity PTM 2}
$\mathcal{P}$ runs in $q(n)$-space, for some polynomial $q: \mathbb{N}\to \mathbb{N}$ with $\deg(q)=d_2$, and
\item \label{thm: completeness complexity PTM 3}
for every $s\in \lbrace 0,1 \rbrace^*$,
$\mathscr{S}_s: \lbrace 0, 1 \rbrace^* \rightarrow [0, 1]$ 
is the probabilistic distribution of the strings that 
$\mathcal{P}$ outputs when applied to input $s$,
\end{enumerate}
then, a term $\underline{\mathcal{P}}$ with type 
$\oc^{\max(d_1,  d_2, 1)+1}\mathbf{S}\multimap \mathbf{S}_{2d_2+1}$ 
exists in {\normalfont\STAP} such that, for every $s\in \lbrace0,1 \rbrace^*$, there exists a surface distribution $\mathscr{D}_s$  satisfying the following conditions:
\begin{enumerate}[i.]
\item \label{enum: completeness 1} 
$\underline{\mathcal{P}}\,  (\oc ^{\max( d_1, d_2, 1)+1} \underline{s} )
 \Rightarrow \mathscr{D}_s$;
\item \label{enum: completeness 2} 
$\mathscr{D}_s(\underline{s'})= \mathscr{S}_s(s')$,
for every $s'\in \lbrace0,1 \rbrace^*$.
\end{enumerate}
\end{theorem}
\begin{proof}[sketch]  
The basic scheme of the proof comes from~\cite{gaboardi2009light}. We first encode natural numbers (with indexed types $\mathbf{N}_i$), all polynomials $p:\mathbb{N}\to \mathbb{N} $, and we define a term $\mathtt{len}_i: \mathbf{S}_i \multimap \mathbf{N}_i$ which, when applied to the encoding $\underline{s}_i$ of a boolean string,  returns $\underline{\vert s \vert}_i$ (where $\vert s\vert$ is the size of  $s$). Then, we firstly represent configurations with indexed types $\mathbf{PTM}_{i}$. Secondly, we encode the transition $\mathtt{tr}: \mathbf{PTM}_i\multimap \mathbf{PTM}_i$ between configurations; $\mathtt{tr}$ relies on the transition functions in~\eqref{equation: transition function PTM}. We also introduce the terms $\mathtt{init}_i: \mathbf{N}_i \multimap \mathbf{PTM}_i$ and 
$\mathtt{in}_i: \mathbf{S} \multimap \mathbf{PTM}_i \multimap \mathbf{PTM}_i$.
The former returns a configuration $\mathtt{C}_0$ having an empty tape with  
$n$ cells, when applied to the numeral $\underline{n}_i$.
The latter fills the empty tape of $\mathtt{C}_0$ with the encodings of the booleans in $s$, whenever applied to the encoding $\underline{s}$ of a boolean string and to $\mathtt{C}_0$. 
Finally, we require the term $\mathtt{ext}_i^{\mathbf{S}}: \mathbf{PTM}_i \multimap \mathbf{S}_i$. It extracts the boolean string on the tape when applied to the 
encoding of a configuration. To sum up, we construct $\underline{\mathcal{P}}$ 
in such a way that, when applied to the encoding of a boolean string $s$:
\begin{itemize}
\item it produces the numerals $\underline{p(\vert s \vert)}$ and $\underline{q(\vert s \vert)}$, where $p: \mathbb{N}\to \mathbb{N}$ is the polynomial bounding the running time of $\mathcal{P}$, and $p: \mathbb{N}\to \mathbb{N}$ is the polynomial bounding the working tape of $\mathcal{P}$;
\item by applying the terms $\mathtt{init}_i$ and $\mathtt{in}_i$, it  constructs the encoding of  the initial configuration  having $q(\vert s \vert)$ cells and  the input string $s$ written on the tape;
\item  it iterates  $p(\vert s \vert)$ times the  transition  $\mathtt{tr}$ to the encoding of the  initial configuration, in order to obtain the encoding of the final configuration;
\item by applying the term $\mathtt{ext}_i^{\mathbf{S}}$  to  the  encoding of the final configuration, it  extracts the encoding of the output string.
\end{itemize}
\end{proof}

\section{\STAP characterizes both \PP and \BPP}
\label{section:STAP characterizes PP and contains BPP}
Previous sections show that \STAP is sound and complete with respect to the functions that a \PTM computes in polynomial time. What about probabilistic polytime complexity classes? 

Let us recall a first basic definition from~\cite{arora2009computational}.
\begin{definition}
[Recognizing a language with error probability $ \epsilon $ by a \PTM] 
\label{definition:Recognizing a language with error epsilon}
Let $\epsilon \in [0,1]$.
Let $T :\mathbb{N} \longrightarrow \mathbb{N}$ be a function. 
Let $L \subseteq \lbrace 0,1 \rbrace^*$ be a language. 
We say that a {\normalfont \PTM} $\mathcal{P}$ 
\emph{recognizes $L$ with error 
probability $\epsilon$  in $T(n)$-time} if:
\begin{itemize}
 \item $x \in L$ implies
 $\mathrm{Pr}[\mathcal{P} \text{ accepts } x] 
 \geq 1-\epsilon$;
 
 \item $x \not \in L$ implies 
 $\mathrm{Pr}[\mathcal{P} \text{ rejects }x]\geq 1-\epsilon$;
 
 \item $\mathcal{P}$ answers ``Accept'' or ``Reject'',
 regardless of its random choices, in at most $T(\vert x \vert)$ steps,
 on every input $x$,
\end{itemize}
where 
$\mathrm{Pr}[\mathcal{P} \text{ accepts } x]$  
(resp.~$\mathrm{Pr}[\mathcal{P} \text{ rejects } x]$) denotes the probability that $\mathcal{P}$ terminates in
an accepting (resp.~rejecting) state on input $x$. 
\end{definition} 
\noindent

Being our goal the characterization of  probabilistic complexity classes 
by means of \STAP, we have to set how a term $ M $, with type in \STAP,
accepts a language. The natural counterpart of 
{\textbf{Definition~\ref{definition:Recognizing a language with error epsilon}}} 
is:

\begin{definition}
[Recognizing a language with error probability $\epsilon$ by \STAP ]
\label{defn: error in STAP} 
Let $\epsilon \in [0,1]$. Let $L \subseteq \lbrace 0,1 \rbrace^*$ be a 
language. By definition, $M : \oc ^n \mathbf{S}\multimap \mathbf{B}$ in 
{\normalfont\STAP}, for some $n \in \mathbb{N}$, 
\emph{recognizes $L $ with error probability $\epsilon$} 
whenever, for every $ x \in\lbrace 0, 1\rbrace^* $, the (unique) surface distribution $ \mathscr{D}_x$ such that $M\, \oc^n\!\underline{x}\Rightarrow \mathscr{D}_x$ satisfies the following conditions:
\begin{enumerate}[(1)]
\item \label{defn: error in STAP 1}
if $x \in L$ then  $\mathscr{D}_x(\underline{0})\geq 1-\epsilon$;
		
\item \label{defn: error in STAP 2}
if $x \not \in L$ then $\mathscr{D}_x(\underline{1})\geq 1-\epsilon$.
\end{enumerate}
\end{definition}

\begin{definition}[The class \PP (from~\cite{arora2009computational})]
\label{defn: The class PP} {\ }
{\normalfont\PP} contains all the languages $ L $ for which a 
{\normalfont \pPTM} $\mathcal{P}$ exists that recognizes $ L $
in $p(n)$-time with error probability $0 \leq \epsilon \leq  \frac{1}{2}$,
where $p$ is a polynomial that depends on $\mathcal{P}$ only.
\end{definition}

\begin{theorem}[\STAP characterizes \PP]
\label{thm: STAP characterizes PP} 
{\normalfont\STAP} is sound and complete w.r.t.~{\normalfont\PP}.
\end{theorem}
\begin{proof}
Concerning the soundness of \STAP w.r.t.~\PP, 
let us fix $ M $ with type in \STAP such that 
$\pi: M \Rightarrow \mathscr{D} $.
{\textbf{Theorem~\ref{thm: polytime soundness complexity}}} 
assures that a {\pPTM} $ \mathcal{P}_{M} $ exists which simulates 
$ \pi $ with a polynomial overhead and with the same probability 
distribution as $ \mathscr{D} $. 
So, if $ M $ recognizes a language $ L $ with error 
probability $ 0 \leq \epsilon \leq \frac{1}{2}$, 
then $ \mathcal{P}_{M} $ does, hence
$ \mathcal{P}_{M} $ is in \PP.
\par
Concerning completeness of \STAP w.r.t.~\PP, 
let $ \mathcal{P} $ be a {\pPTM} in \PP . The proof is the one for  
{\textbf{Theorem~\ref{thm: completeness complexity PTM}}}, 
but we have to represent a \pPTM that decides a problem instead of one that
computes a function. 
W.l.o.g.,~we assume that a final state is either accepting or rejecting. 
Then, we simply replace the term 
$\mathtt{ext}^{\mathbf{B}}_i: \mathbf{PTM}_i \multimap \mathbf{B}$,
which extracts the final state from the final configuration (see~\cite{gaboardi2009light}),
for $\mathtt{ext}^{\mathbf{S}}_i: \mathbf{PTM}_i \multimap \mathbf{S}_i$,
which extracts the output string from the final configuration.
(We recall that $\mathbf{PTM}_i$ is the indexed type for configurations.)
So, $ \underline{\mathcal{P}} $ accepts a 
language $ L $ with the same probability error 
$ 0\leq \epsilon \leq \frac{1}{2} $ as $ \mathcal{P} $.
\end{proof}
\vspace{.5\baselineskip}
\noindent
Here above, \PP is instance of a general notion, formalized in 
{\textbf{Definition~\ref{definition:Recognizing a language with error epsilon}}}.
However, the interval that the  error  probability identifying \PP belongs to 
allows for a further definition of this class, equivalent to
{\textbf{Definition~\ref{defn: The class PP}}}.

\begin{definition}
	[\PP recognizes by majority]
\label{definition:PP and recognizing by majority}
{\normalfont\PP} contains all the languages $ L $ for which a 
{\normalfont \pPTM} $\mathcal{P}$ exists such that,
for every $ x\in\lbrace 0, 1\rbrace^* $, 
both the following points
{{{(\ref{definition:PP and recognizing by majority 1})}}}
and~{{{(\ref{definition:PP and recognizing by majority 2})}}}
hold:
\begin{enumerate}[(1)]
	\item\label{definition:PP and recognizing by majority 1} 
	if $x \in L$, then 
	$\mathrm{Pr}[\mathcal{P} \textrm{ accepts } x] 
	   \geq \mathrm{Pr}[\mathcal{P} \textrm{ rejects } x]$;
	\item\label{definition:PP and recognizing by majority 2} 
	if $x \not\in L$, then 
    $\mathrm{Pr}[\mathcal{P} \textrm{ rejects } x] 
       \geq \mathrm{Pr}[\mathcal{P} \textrm{ accepts } x]$.
\end{enumerate}
\end{definition}

\begin{definition}[\STAP recognizes by majority] 
\label{defn: STAP by majority}  
Let $L \subseteq \lbrace 0,1 \rbrace^*$ be a language.
Let $M$  be a term with type $\oc ^n \mathbf{S}\multimap \mathbf{B}$ 
in {\normalfont\STAP}, for some $n \in \mathbb{N}$. 
We say that $M$ \emph{accepts $L$ by majority}  whenever, for every $ x \in\lbrace 0, 1\rbrace^* $, the (unique) surface distribution $ \mathscr{D}_x$ such that $M\, \oc^n\!\underline{x}\Rightarrow \mathscr{D}_x$ satisfies the following conditions:
\begin{enumerate}[(1)]
\item\label{defn: STAP by majority 1} 
if $x \in L$ then 
$\mathscr{D}_x(\underline{0}) \geq \mathscr{D}_x(\underline{1})$;

\item\label{defn: STAP by majority 2} 
if $x \not \in L$ then 
$\mathscr{D}_x(\underline{1})\geq \mathscr{D}_x(\underline{0})$.
\end{enumerate}
\end{definition}
\vspace{.5\baselineskip}
\noindent
A proof analogous to the one for 
{\textbf{Theorem}}~\ref{thm: STAP characterizes PP}
exists for the following theorem which, however, refers to
{\textbf{Definition}}~\ref{definition:PP and recognizing by majority} 
and 
{\textbf{Definition}}~\ref{defn: STAP by majority}:

\begin{theorem}
[\STAP characterizes $\mathsf{PP}$  by majority] 
\label{thm:STAP characterizes PP with a recognition by majority}
{\normalfont\STAP} is sound and complete w.r.t.~{\normalfont\PP}.
\end{theorem}

Let us now turn our attention to the relation between \STAP and \BPP.

\begin{definition}[The class \BPP (from~\cite{arora2009computational})]
\label{defn: the class BPP} {\ }
{\normalfont\BPP} is the class of all languages $ L $ for which a 
{\normalfont\pPTM} $\mathcal{P}$ exists that recognizes $ L $
in $p(n)$-time with error probability $0 \leq \epsilon <  \frac{1}{2}$,
and $p$ is a polynomial that depends on $\mathcal{P}$ only.
\end{definition}

\begin{remark}
The value $\epsilon$ cannot be equal to $\frac{1}{2}$ in \BPP. 
Due to this restriction the error probability can be made 
exponentially small at the cost of a polynomial  
slowdown~\cite{sipser2012introduction}. This is why 
\BPP is widely considered as the class capturing \emph{efficient} 
(probabilistic) computations.
\lipicsEnd
\end{remark}

\begin{theorem}[\STAP characterizes \BPP]
\label{thm: STAP characterizes BPP} 
{\normalfont\STAP} is sound and complete w.r.t.~{\normalfont\BPP}.
\end{theorem}
\begin{proof}
It is like the proof of 
{\textbf{Theorem~\ref{thm: STAP characterizes PP}}}.
\end{proof}

As far as we know,
no alternative definition of \BPP, analogous to 
{\textbf{Definition~\ref{definition:PP and recognizing by majority}}} and referring to an error probability implicitly, exists. Our feeling is that one can achieve a better insight on this class  by moving to a semantic framework. This is where $\STAP$ can play a role. One can indeed exploit denotational semantics, 
available for deductive systems based on \LL, to semantically characterize 
probabilistic computational complexity classes which, currently,
\STAP characterizes operationally.  {\textbf{Conclusions}} elaborate slightly on this.

\section{Conclusions}
\label{section:Conclusions}
We illustrate how the relevant features of \STAP, i.e.~both its
polynomially costing non-deterministic normalization, with a natural 
probabilistic interpretation, and its connections with \LL structural proof-theory, 
can be the base for generalizing known results or shading some light on open issues.
\par
We think that \STAP can be used to improve  
known characterizations of the class \NPTIME, as given in  $\mathsf{STA}_+$ by Marion et al.~\cite{gaboardi2008soft}.
We recall that $\mathsf{\STA}_+$ is \STA extended with a \emph{sum-rule}.
That sum-rule gives a type to a choice operator $M + N$, i.e. to 
an oracle that autonomously ``decides'' when reducing to either 
$M$ or $N$.  
The normalization steps associated with the sum-rule suffer the typical drawback  of additives in deductive systems based on \LL: the cost of 
normalizing terms with a type in $\mathsf{STA}_+$ may result in an exponential
blow up. To recover \NPTIME soundness, the normalization of
terms with a type in $\mathsf{STA}_+$ must be a variant of the 
leftmost outermost strategy, delaying substitutions as long as 
possible. 
By contrast, thanks to the inherently linear nature of non-determinism in \STAP, arising from a careful managing of context-sharing in Linear additives,  \STAP enjoys a \emph{strong} polynomial time normalization. Therefore, non-deterministic Linear additives can  be employed to make the characterization  of \NPTIME free of any explicit  reference to  reduction strategies.

We also think that \STAP, which stems from proof-theoretical principles, will be useful to 
address the problem of characterizing implicitly the class \BPP. As pointed out  
also in~\cite{dal2015higher}, characterizing \BPP by purely syntactical means 
is far from obvious, for it boils down to identify some structural invariant 
that allows to recognize a language with an error probability 
\emph{strictly smaller} than $ \frac{1}{2} $.
Given that invariant, possibly captured inside an inductively defined formal
system, one could be able, in principle, to enumerate all the algorithms of
\BPP. 
\par
Denotational semantics can be a way to suggest such a
structural invariant, and \STAP can play a crucial role.
\STAP is a probabilistic type-theoretical formulation of \SLL, a 
subsystem of \LL capturing the complexity class \PTIME.
Probabilistic denotational models for \LL exist, e.g.~Probabilistic Coherence Spaces $\mathsf{PCoh}$~\cite{danos2011probabilistic} 
or Weighted Relational Semantics~\cite{laird2013weighted}, so they can be easily adapted to \STAP.
What we are looking for in these models is a probabilistic version of the notion of  \emph{obsessionality}~\cite{laurent2006obsessional}, an 
invariant found in relational models for \SLL, and used to characterize 
\PTIME denotationally.

\bibliographystyle{plainurl}
\bibliography{main}

\appendix
\section{Confluence for $\plamb$}\label{app3}
In this section we prove that the   probabilistic multi-step reduction $\Rightarrow$ defined in \textbf{Figure}~\ref{fig: operational semantics for staoplus} is confluent, that is,  each term of $\plamb$ can be associated with at most one surface distribution. This property is shown by adapting the techniques in Dal Lago and Toldin~\cite{dal2015higher}. 

The first step  is to prove that   $\rightarrow$ enjoys a strong confluence property for $\plamb$:
\begin{lemma} \label{lem: simpson sub lemma} Let $M, N\in \plamb$:
\begin{enumerate}[(1)]
\item \label{enum: simpson conf 1} If $M \rightarrow M', M''$ then $M\lbrace  N/x\rbrace \rightarrow M'\lbrace  N/x\rbrace , M'' \lbrace  N/x\rbrace $
\item \label{enum: simpson conf 2}If $N \rightarrow N', N''$ and $x$ is linear in $M$ then $M[N/x]\rightarrow M[N'/x],M[N''/x]$.
\end{enumerate}
\end{lemma}
\begin{proof}
Easy induction on the structure of $M$.
\end{proof}
\begin{lemma} \label{lem: one-reduction step confluence 1} Let $M\in \plamb$. If $M \rightarrow M'$ and $M \rightarrow M''$, with $M'$ and $M''$ distinct, then there exists a term $N$ such that $M' \rightarrow N$ and $M'' \rightarrow N$.
\end{lemma}
\begin{proof}
By induction on the structure of $M$. We just consider the most interesting cases. If $M= (\lambda x. P)Q\rightarrow P[Q/x]=M'$,  then either $M''= (\lambda x. P')Q$ with $P \rightarrow P'$ or $M''= (\lambda x. P)Q'$ with $Q\rightarrow Q'$. Since $M$ is $s$-linear, $x$ is $s$-linear in $P$ and hence $x$ does not lie within the scope of a $\mathtt{d}$-operator. This means that $P[Q/x]= P\lbrace Q/x \rbrace$ by definition.  In the first case, we have $M' \rightarrow P'[Q/x]$ by \textbf{Lemma}~\ref{lem: simpson sub lemma}.\ref{enum: simpson conf 1} and also $M'' \rightarrow P'[Q/x]$. In the second case, we have $M' \rightarrow P[Q'/x]$ by \textbf{Lemma}~\ref{lem: simpson sub lemma}.\ref{enum: simpson conf 2}, and also $M'' \rightarrow P[Q'/x]$. Similarly, if $M= (\lambda \oc x. P)\oc Q  \rightarrow P\lbrace \oc Q /x \rbrace =M'$  then the only case is  $M''= (\lambda \oc x. P')\oc Q$ where $P \rightarrow P'$, since reduction is forbidden in $Q$.  By  \textbf{Lemma}~\ref{lem: simpson sub lemma}.\ref{enum: simpson conf 1},  $M' \rightarrow  P'\lbrace \oc Q/x \rbrace$,  and also $M'' \rightarrow  P'\lbrace \oc Q/x\rbrace $. Last, we consider the case where $M= \mathtt{copy}^{U} \, V \mathtt{\ as \ }x_1, x_2 \mathtt{ \ in \ }\langle N_1, N_2\rangle $,  $M'= \langle N_1[V/x_1], N_2[V/x_2]\rangle $, and $M''= \mathtt{copy}^{U} \, V \mathtt{\ as \ }x_1, x_2 \mathtt{ \ in \ }\langle N'_1, N_2\rangle $. Since $M$ is $s$-linear, $x_1$ is $s$-linear in $N_1$ and hence $x$ does not lie within the scope of a $\mathtt{d}$-operator. This means that $N_1[V/x]= N_1 \lbrace V/x \rbrace$ by definition.  Then $M' \rightarrow \langle N'_1[V/x_1], N_2[V/x_2]\rangle $ by \textbf{Lemma}~\ref{lem: simpson sub lemma}.\ref{enum: simpson conf 1} and also $M'' \rightarrow \langle N'_1[V/x_1], N_2[V/x_2]\rangle $.
\end{proof}
\begin{lemma} \label{lem: one-reduction step confluence 2} Let $M\in \plamb$.  If $M \rightarrow M'_1, M'_2$ and $M \rightarrow M''$, with $M'_1$ and $M'_2$ distinct, then there exist terms $N_1$ and  $N_2$ such that $M'_1 \rightarrow N_1$, $M'_2\rightarrow N_2$ and $M''\rightarrow N_1, N_2$.
\end{lemma}
\begin{proof}
The proof is by induction on the structure of $M$. The only possible situation is when both the  surface reductions $M \rightarrow M'_1, M'_2$ and $M \rightarrow M''$ are applied in surface  contexts $\mathcal{C}\neq [\cdot]$, and we proceed by case analysis. We just consider a possible case. Suppose $M= PQ \rightarrow P'_1Q , P'_2Q$, where $P'_1Q=M'_1$ and $P'_2Q=M'_2$. Then either $M''= P''Q$, where $P \rightarrow P''$, or $M''= PQ''$, where $Q \rightarrow Q''$. In the first case we apply the induction hypothesis on $P \rightarrow P'_1, P'_2$ and $P \rightarrow P''$ and we get that there exist $R_1$ and $R_2$  such that $P'_1 \rightarrow R_1$, $P'_2 \rightarrow R_2$ and $P'' \rightarrow R_1, R_2$, so that $P'_1Q \rightarrow R_1Q$, $P'_2Q\rightarrow R_2Q$ and $P''Q\rightarrow R_1Q, R_2Q$. In the second case, we have $P'_1Q \rightarrow P'_1Q''$,  $P'_2 Q \rightarrow P'_2Q''$ and $PQ''\rightarrow P'_1Q'', P'_2Q''$. 

\end{proof}

\begin{lemma}\label{lem: one-reduction step confluence 3}   Let $M\in \plamb$. If $M \rightarrow M'_1, M'_2$ and $M \rightarrow M''_1, M''_2$, with $M'_1$, $M'_2$, $M''_1$, $M''_2$ all distinct,  then there exist $N_1, N_2, N_3, N_4$ such that $M'_1 \rightarrow N_1, N_2$, $M'_2 \rightarrow N_3, N_4$ and $\exists i\in \lbrace 1,2\rbrace$ such that $M''_i \rightarrow N_1, N_3$ and $M''_{3-i}\rightarrow N_2, N_4$. 
\end{lemma}
\begin{proof}
The proof is by induction on the structure of $M$. The only possible situation is when both the  surface reductions $M \rightarrow M'_1, M'_2$ and $M \rightarrow M''_1, M''_2$ are applied in surface  contexts $\mathcal{C}\neq [\cdot]$, and we proceed by case analysis. We just consider a possible case. Suppose $M= PQ \rightarrow P'_1Q , P'_2Q$, where $M'_1=P'_1Q$ and $M'_2= P'_2Q$. Then either $M''_1= P''_1Q$, $M''_2=P''_2Q$ or $M''_1= PQ''_1$, $M''_2=PQ''_2$. In the first case we apply the induction hypothesis on $P \rightarrow P'_1, P'_2$ and $P \rightarrow P''_1, P''_2$ and we have that there exist $R_1, R_2, R_3, R_4$ such that $P'_1 \rightarrow R_1, R_2$,  $P'_2 \rightarrow R_3, R_4$ and $\exists i$ such that $P''_i \rightarrow R_1, R_3$ and $P''_{3-i}\rightarrow R_2, R_4$. Then,  we have $P'_1Q \rightarrow R_1Q, R_2Q$, $P'_2Q \rightarrow R_3Q, R_4Q$, $P''_iQ \rightarrow R_1Q, R_3Q$, and $P''_{3-i}Q \rightarrow R_2Q, R_4Q$. In the second case we have  $P'_1Q\rightarrow P'_1Q''_1, P'_1Q''_2$, $P'_2Q \rightarrow P'_2Q''_1, P'_2Q''_2$, $PQ''_1 \rightarrow P'_1Q''_1, P'_1Q''_2$, and $PQ''_2 \rightarrow P'_1Q''_2, P'_2Q''_2$. 

\end{proof}

The next step is to introduce  a probabilistic multi-step reduction relation $\Rrightarrow$ which is \enquote{laxer} than $\Rightarrow$, i.e.~such that ${\Rightarrow}  \subseteq  {\Rrightarrow}$.

\begin{definition}[Multi-step reduction $\Rrightarrow$]\label{defn: multistep probabilsitic operational semantics for surface reduction}{\ }
\begin{itemize}
\item A \emph{term distribution} is a probability distribution over $\plamb$, i.e.~a function $\mathscr{D}:\plamb \longrightarrow [0, 1]$ such that $\textstyle \sum_{M\in \plamb} \mathscr{D}(M)=1$.
\item The \emph{multi-step reduction $\Rrightarrow $} is the  relation    between terms of $\plamb $   and term distributions,  defined by the rules in \normalfont{\textbf{Figure}}~\ref{fig: multistep probabilsitic operational semantics for surface reduction}. Derivations of $M \Rightarrow \mathscr{D}$ are ranged over by $\pi, \rho$.   
\item The  \textit{size}  $\vert \pi \vert$ of a derivation $\pi: M \Rrightarrow \mathscr{D}$ is  $0$ if   $\pi$ is $t1$, and $\s{\pi}\triangleq\max(\s{\pi_1}, \s{\pi_2})+1$ if $\pi$ is $t2$ with premises  $M \rightarrow M_1, M_2$, $\pi_1:M_1 \Rrightarrow \mathscr{D}_1$ and $\pi_2:M_2 \Rrightarrow \mathscr{D}_2$. Henceforth, with a little abuse of notation, we  shall  write $\s{ M \Rrightarrow \mathscr{D}}$ in place of $\s{\pi}$, whenever $\pi: M \Rrightarrow \mathscr{D}$.
\end{itemize}
\end{definition}

\begin{figure}[t]
\begin{mathpar}
 \inferrule*[Right=$t1$]{M \in \plamb}{M \Rrightarrow M}\and
\inferrule*[Right=$t2$]{M \rightarrow M_1, M_2 \\ M_1 \Rrightarrow \mathscr{D}_1 \\ M_2 \Rrightarrow \mathscr{D}_2}{M \Rrightarrow  \textstyle{\frac{1}{2}}\cdot \mathscr{D}_1 + \textstyle{\frac{1}{2}}\cdot \mathscr{D}_2}
\end{mathpar}
\caption{Multi-step reduction $\Rrightarrow$ for $\plamb$.}
\label{fig: multistep probabilsitic operational semantics for surface reduction}
\end{figure}
Notice that the only difference between the relations $\Rightarrow$ and $\Rrightarrow$ is that $s1$ applies to surface normal forms only, while  $t1$ applies to all terms.  The following states that ${\Rightarrow}  \subseteq  {\Rrightarrow}$:
\begin{lemma} \label{lem: rightarrow implies rrightarrow}If   $\pi : M\Rightarrow \mathscr{D} $ then  there exists a derivation $\pi'$  such that $\pi':M \Rrightarrow \mathscr{D}$ and $\s{\pi}=\s{\pi'}$.
\end{lemma}
 
Confluence for $\Rightarrow$  follows directly from two technical results about  $\Rrightarrow$.    
 
\begin{lemma} \label{lem: to prove confluence} Let $M\in \plamb$. Let $M \Rrightarrow \mathscr{D}$ be such that $\mathscr{D}=  p_1\cdot N_1+ \ldots + p_n\cdot N_n$, and let    $N_i \Rrightarrow \mathscr{E}_i$ for all $i \leq n $. Then: 
\begin{enumerate}[(1)]
\item $M \Rrightarrow \sum_{i=1}^n p_i \cdot \mathscr{E}_i$ 
\item $\vert M \Rrightarrow \sum_{i=1}^n p_i\cdot \mathscr{E}_i \vert \leq \vert M \Rrightarrow \mathscr{D} \vert + \max_{i=1}^n \vert N_i \Rrightarrow \mathscr{E}_i  \vert$.
\end{enumerate}
\end{lemma}
\begin{proof}
The proof is by induction on the structure of the derivation of $M \Rrightarrow \mathscr{D}$, and follows exactly~\cite{dal2015higher}.
\end{proof}

\begin{lemma}\label{lem: strenghtening of confluence} Let $M\in \plamb$.  If $M \Rrightarrow \mathscr{D}$ and $M \Rrightarrow \mathscr{E}$, where $\mathscr{D}= p_1\cdot P_1+ \ldots+ p_n \cdot P_n$ and $\mathscr{E}=q_1\cdot Q_1+ \ldots + q_m\cdot Q_m$,  then there exist $ \mathscr{L}_1, \ldots \mathscr{L}_n$ and $ \mathscr{F}_1\ldots \mathscr{F}_m$ such that:
\begin{itemize}
\item  $P_i \Rrightarrow \mathscr{L}_i$ and $Q_j \Rrightarrow \mathscr{F}_j$, for all $i \leq n$,  $j \leq m$; 
\item $\max_{i=1}^n \vert P_i \Rrightarrow \mathscr{L}_i \vert \leq \vert M \Rrightarrow \mathscr{E} \vert $ and  $\max_{j=1}^m \vert Q_j \Rrightarrow \mathscr{F}_j \vert \leq \vert M \Rrightarrow \mathscr{D} \vert$;
\item $\sum_{i=1}^n p_i \cdot \mathscr{L}_i= \sum_{j=1}^m q_j \cdot \mathscr{F}_j$.
\end{itemize}
\end{lemma}
\begin{proof}
By induction on $\vert M \Rrightarrow \mathscr{D}  \vert + \vert M \Rrightarrow \mathscr{E} \vert$.  If  one of the derivations ends with $t1$ then there is nothing to prove. 
 Otherwise, both derivations $M \Rrightarrow \mathscr{D}$ and $M \Rrightarrow \mathscr{E}$  end with the rule $t2$
\begin{center}
\begin{tabular}{cc}
\def\defaultHypSeparation{\hskip .1in}
\AxiomC{$M \rightarrow M_1, M_2$}
\AxiomC{$M_1 \Rrightarrow \mathscr{D}_1$}
\AxiomC{$M_2 \Rrightarrow \mathscr{D}_2$}
\RightLabel{$t2$}
\TrinaryInfC{$M \Rrightarrow \frac{1}{2}\cdot \mathscr{D}_1+ \frac{1}{2}\cdot \mathscr{D}_2$}
\DisplayProof
\\ \\ \\
\def\defaultHypSeparation{\hskip .1in}
\AxiomC{$M \rightarrow N_1, N_2$}
\AxiomC{$N_1 \Rrightarrow \mathscr{E}_1$}
\AxiomC{$N_2 \Rrightarrow \mathscr{E}_2$}
\RightLabel{$t2$}
\TrinaryInfC{$M \Rrightarrow \frac{1}{2}\cdot \mathscr{E}_1+ \frac{1}{2}\cdot \mathscr{E}_2$}
\DisplayProof
\end{tabular}
\end{center}
Clearly, if $M_1, M_2$ is equal to $N_1, N_2$ (modulo sort) then we apply the induction hypothesis and we are done. So let us suppose that $M_1, M_2$ and $N_1, N_2$ are different. We have four cases:
\begin{itemize}
\item If $M_1=M_2$ and $N_1= N_2$ then by \textbf{Lemma}~\ref{lem: one-reduction step confluence 1}  there exists $L$ such that $M_1 \rightarrow L$ and $N_1 \rightarrow L$. By using the rule $t1$ we  get $L \Rrightarrow L$, so $M_1 \Rrightarrow L$. By induction hypothesis on $M_1 \Rrightarrow \mathscr{D}_1$ and $M_1 \Rrightarrow L$ there exist $\mathscr{L}_1, \ldots, \mathscr{L}_n$ and $\mathscr{K}$ such that,  for all $i \leq n$, $P_i \Rrightarrow \mathscr{L}_i$, $L \Rrightarrow \mathscr{K}$, $\max_{i=1}^n (\vert P_i \Rrightarrow \mathscr{L}_i\vert)\leq \vert M \Rrightarrow L\vert$, $\vert L \Rrightarrow \mathscr{K}\vert \leq \vert M \Rrightarrow \mathscr{D} \vert $, and $\sum_{i=1}^n p_i\cdot\mathscr{L}_i= \mathscr{K}$. Similarly, we have that there exists $\mathscr{F}_1, \ldots, \mathscr{F}_m, \mathscr{H}$ such that, for all $i \leq m$, $Q_i \Rrightarrow \mathscr{F}_i$, $L \Rrightarrow \mathscr{H}$, $\max_{i=1}^m (\vert Q_i \Rrightarrow \mathscr{F}_i\vert)\leq \vert M \Rrightarrow L \vert$, $\vert L \Rrightarrow \mathscr{H } \vert\leq \vert M \Rrightarrow \mathscr{E} \vert $, and $\sum_{i=1}^m q_i\cdot \mathscr{F}_i= \mathscr{H}$. We obtain $\vert L \Rrightarrow \mathscr{K} \vert + \vert L \Rrightarrow \mathscr{H} \vert \leq \vert M \Rrightarrow \mathscr{D} \vert + \vert M \Rrightarrow \mathscr{E}\vert$. Let $\mathscr{K}= r_1\cdot R_1+ \ldots + r_h\cdot R_h$ and $\mathscr{H}= s_1\cdot S_1+ \ldots + s_k\cdot S_k$. We apply the induction hypothesis and we obtain that there exist $\mathscr{R}_1, \ldots, \mathscr{R}_h, \mathscr{S}_1, \ldots, \mathscr{S}_k$ such that $R_i \Rrightarrow \mathscr{R}_i$ and  $S_j\Rrightarrow \mathscr{S}_j$ for all $i\leq h$ and $j\leq k$. Moreover, $\max_{i=1}^h (\vert R_i \Rrightarrow \mathscr{R}_i \vert)\leq \vert L \Rrightarrow \mathscr{H}\vert$, $\max_{k=1}^k(\vert S_j \Rrightarrow \mathscr{S}_j \vert)\leq \vert L \Rrightarrow \mathscr{K}\vert$, and $\sum_{i=1}^h r_i\cdot \mathscr{R}_i= \sum_{j=1}^k s_j\cdot \mathscr{S}_j$. Notice that the cardinality of $\mathscr{D}$ and $\mathscr{K}$ may differ but for sure they have the same terms with non zero probability. Similar, $\mathscr{E}$ and $\mathscr{H}$ have the same terms with non zero probability.
By using \textbf{Lemma}~\ref{lem: to prove confluence} and using the transitive property of equality we obtain that $\sum_{i=1}^n p_i \cdot \mathscr{R}_i = \sum_{i=1}^n r_i\cdot \mathscr{R}_i=\sum_{j=1}^m  s_j\cdot \mathscr{S}_j = \sum_{j=1}^m q_j\cdot \mathscr{S}_j$. Moreover, we have
\allowdisplaybreaks
\begin{align*}
&\max_{i=1}^{n} (\vert P_i \Rrightarrow \mathscr{R}_i\vert )\leq \vert L \Rrightarrow \mathscr{H}\vert \leq \vert M \Rrightarrow \mathscr{E}\vert\\
& \max_{j=1}^m (\vert Q_j \Rrightarrow \mathscr{S}_j \vert )\leq \vert L \Rrightarrow \mathscr{K} \vert \leq \vert M \Rrightarrow \mathscr{D}\vert  .
\end{align*}
\item If $M_1 \not = M_2$ and $N_1= N_2$ then by \textbf{Lemma}~\ref{lem: one-reduction step confluence 2}  there exists $L_1, L_2$ such that $M_1 \rightarrow L_1$, $M_2 \rightarrow L_2$ and $N_1 \rightarrow L_1, L_2$. W.l.o.g. we can assume that $\mathscr{D}_1= 2p_1\cdot P_1+ \ldots + 2p_{o-1}\cdot P_{o-1}+ p_o\cdot P_{o}+ \ldots + p_t\cdot P_t$ and $\mathscr{D}_2= p_o\cdot P_o+ \ldots + p_{h}\cdot P_{h}+ 2p_{t+1}\cdot P_{t+1}+ \ldots + 2p_n\cdot  P_n$ where $1 \leq o \leq t \leq n$.By using the induction rule, we associate with every $L_i$ a distribution $\mathscr{P}_i$ such that $L_1 \Rrightarrow \mathscr{P}_1$ and $L_2 \Rrightarrow \mathscr{P}_2$. Let $\mathscr{P}_1= r_1\cdot R_1 + \ldots + r_h\cdot R_h$ and $\mathscr{P}_2= s_1 \cdot S_1 + \ldots + s_k\cdot S_k$. So, we have, for all $i$, $M_i \Rrightarrow \mathscr{D}_i$ and $M_i \Rrightarrow \mathscr{P}_i$, $N_1 \Rrightarrow \mathscr{E}$ and $N_1 \Rrightarrow \frac{1}{2}\cdot \mathscr{P}_1+ \frac{1}{2}\cdot \mathscr{P}_2$. By applying the induction hypothesis on all the three cases we have that there exist $\mathscr{L}_1,\ldots, \mathscr{L}_n, \mathscr{F}_1, \ldots, \mathscr{F}_m, \mathscr{K}, \mathscr{H}, \mathscr{R}, \mathscr{S}$ such that $P_1\Rrightarrow \mathscr{L}_1, \ldots, P_n \Rrightarrow \mathscr{L}_n $, $Q_1 \Rrightarrow \mathscr{F}_1$, \ldots, $P_m \Rrightarrow \mathscr{F}_m$,  $L_1 \Rrightarrow \mathscr{K}$, $L_2 \Rrightarrow \mathscr{H}$, $L_1 \Rrightarrow \mathscr{R}$, and $L_2 \Rrightarrow \mathscr{S}$. Moreover:
\begin{enumerate}
\item $\max_{1 \leq i \leq t} (\vert P_i \Rrightarrow \mathscr{L}_i \vert)\leq \vert M_1 \Rrightarrow \mathscr{P}_1\vert$, $\vert L_1 \Rrightarrow \mathscr{K}\vert \leq \vert M_1 \Rrightarrow \mathscr{D}_1\vert$, and $\sum_{i=1}^{o-1} 2p_i \cdot \mathscr{L}_i+ \sum_{i=o}^t p_i\cdot  \mathscr{L}_i= \mathscr{K}$.
\item $\max_{o \leq i \leq n} (\vert P_1 \Rrightarrow \mathscr{L}_i \vert)\leq \vert M_2 \Rrightarrow \mathscr{P}_2\vert$, $\vert L_2 \Rrightarrow \mathscr{H}\vert \leq \vert M_2 \Rrightarrow \mathscr{D}_2\vert$, and $\sum_{i=o}^{t} p_i \cdot \mathscr{L}_i+ \sum_{i=t+1}^n 2p_i \cdot \mathscr{L}_i= \mathscr{H}$.
\item $\max_{j=1}^m (\vert Q_j \Rrightarrow \mathscr{F}_j \vert)\leq \vert N_1 \Rrightarrow \frac{1}{2}\cdot \mathscr{P}_1+ \frac{1}{2}\cdot \mathscr{P}_2\vert $, $\max (\vert L_1 \Rrightarrow \mathscr{R} \vert, \vert L_2 \Rrightarrow \mathscr{S} \vert ) \leq \vert N_1 \Rrightarrow \mathscr{E} \vert$, and $\sum_{j=1}^m q_j \cdot \mathscr{F}_j = \frac{1}{2}\cdot \mathscr{R}+ \frac{1}{2}\cdot \mathscr{S}$.
\end{enumerate}
Notice that $\vert L_1 \Rrightarrow \mathscr{R}\vert + \vert L_1 \Rrightarrow \mathscr{K}\vert< \vert M \Rrightarrow \mathscr{D}\vert + \vert M \Rrightarrow \mathscr{E}\vert$. Moreover, notice also that the following inequality holds: $\vert L_2 \Rrightarrow \mathscr{S}\vert + \vert L_2 \Rrightarrow \mathscr{H} \vert < \vert M \Rrightarrow \mathscr{D}\vert + \vert M \Rrightarrow \mathscr{E} \vert$.
We are allowed to apply, again,  induction hypothesis and have a confluent distribution for both cases.  \textbf{Lemma}~\ref{lem: to prove confluence} then allows us to connect the first two main derivations and by transitivive property of equality we have the thesis.
\item The case $M_1= M_2$ and $N_1 \not = N_2$ and the case $M_1 \not = M_2$ and $N_1 \not = N_2$ are proven similarly by using, respectively, \textbf{Lemma}~\ref{lem: one-reduction step confluence 2} and \textbf{Lemma}~\ref{lem: one-reduction step confluence 3}.
\qedhere
\end{itemize}
\end{proof} 

Finally, we are ready for the following proof:
\begin{proof}[Proof of {\normalfont\textbf{Theorem}~\ref{thm: confluence for lamb}}]
Since $\Rightarrow \, \subseteq \, \Rrightarrow$, we have $M \Rrightarrow \mathscr{D}$ and $M \Rrightarrow \mathscr{E}$. By \textbf{Lemma}~\ref{lem: strenghtening of confluence},  $\mathscr{D}= \mathscr{E}$.
\end{proof}
\section{Proofs of Section~\ref{sec4}}
\begin{proof}[Proof of {\normalfont\textbf{Lemma}~\ref{lem: properties of weight}}]
By induction on the structure of $\mathcal{D}$. Point~\textsf{\textbf{\ref{enum: properties of weight 1}}} and point~\textsf{\textbf{\ref{enum: properties of weight 3}}} are straightforward. Concerning point~\textsf{\textbf{\ref{enum: properties of weight 2}}}, we consider the most interesting case where $\mathcal{D}$ has been obtained from a derivation $\mathcal{D}'$ by applying the rule $sp$ with $sp$-rank $k$. By using the induction hypothesis, we have: 
\begin{equation*}
\begin{aligned}
\mathrm{w}(\mathcal{D}, r)&= r \cdot( \mathrm{w}(\mathcal{D}', r)+k)+1 \\
&\leq  r \cdot (r^{\mathrm{d}(\mathcal{D}')} \cdot \mathrm{w}(\mathcal{D}', 1)+k) +1\\
&\leq  r \cdot (r^{\mathrm{d}(\mathcal{D}')} \cdot \mathrm{w}(\mathcal{D}', 1)+ r^{\mathrm{d}(\mathcal{D}')}\cdot k) +r^{\mathrm{d}(\mathcal{D}')+1}\\
&\leq   r^{\mathrm{d}(\mathcal{D}')+1} \cdot (\mathrm{w}(\mathcal{D}', 1) +k+1)  = r^{\mathrm{d}(\mathcal{D})} \cdot \mathrm{w}(\mathcal{D}, 1)  .
\end{aligned}
\end{equation*}
\end{proof}

The following lemmas can be easily proved by inspecting the rules of \STAP.

\begin{lemma}[Generation] \label{lem: generation} {\ }
\begin{enumerate}
\item \label{enum: generation for STA lin absts  types} If $\mathcal{D}\triangleleft \Gamma \vdash \lambda  x. M: \sigma$  then $\sigma= \forall \vec{\alpha}.( (A \multimap B)\langle D_1/\beta_1, \ldots, D_n/\beta_n \rangle)$ and $\mathcal{D}$ is some $\mathcal{D}'\triangleleft \Gamma, x: A \vdash M': B$  followed by  ${\multimap} \mathrm{I}l$  and a sequence of  $\mathrm{\forall I}$, $\mathrm{\forall E}$,  and $m$ where  $\vec{\alpha}=\alpha_1, \ldots, \alpha_k$, for some $k\geq 0$.
\item \label{enum: generation for STA exp absts  types} If $\mathcal{D}\triangleleft \Gamma \vdash \lambda \oc x. M: \sigma$  then $\sigma= \forall \vec{\alpha}.( (\oc \tau \multimap A)\langle D_1/\beta_1, \ldots, D_n/\beta_n \rangle)$ and $\mathcal{D}$ is some $\mathcal{D}'\triangleleft \Gamma, x:\oc  \tau \vdash M': A$  followed by  ${\multimap} \mathrm{I}e$ and a sequence of  $\mathrm{\forall I}$, $\mathrm{\forall E}$,  and $m$ where  $\vec{\alpha}=\alpha_1, \ldots, \alpha_k$, for some $k\geq 0$.
\item \label{enum: generation for STA app types} If $\mathcal{D}\triangleleft \Gamma \vdash MN: \sigma$ then $\sigma=\forall \vec{\alpha}. (A\langle D_1/\beta_1, \ldots, D_n/\beta_n \rangle) $ and $\mathcal{D}$  is some $\mathcal{D}'\triangleleft \Gamma' \vdash M': \tau \multimap A$ and $\mathcal{D}''\triangleleft \Gamma'' \vdash N': \tau$   followed by  ${\multimap} \mathrm{E}$  and a sequence of  $\mathrm{\forall I}$, $\mathrm{\forall E}$, and $m$, where $\vec{\alpha}=\alpha_1, \ldots, \alpha_k$ and   $k\geq 0$.
\item\label{enum: generation for STA copy}  If $\mathcal{D}\triangleleft \Gamma \vdash \mathtt{copy}^V \, N \mathtt{\ as \ }x_1,x_2 \mathtt{\ in \ } \langle M_1, M_2\rangle : \sigma$, then $\sigma = B_1\with B_2$  and $\mathcal{D}$ is $\mathrm{\with I}$ followed by a sequence of applications of the rule $m$.

\item \label{enum: generation for STA pi} If $\mathcal{D}\triangleleft \Gamma \vdash \mathtt{proj}(M):\sigma$ then $\sigma=\forall \vec{\alpha}.( B_i\langle D_1/\beta_1, \ldots, D_n/\beta_n \rangle)$, and  $\mathcal{D}$ is  $\mathcal{D}' \triangleleft \Gamma' \vdash M': B_1\with B_2$ followed by $\mathrm{\with E}$ and a sequence of  $\mathrm{\forall I}$, $\mathrm{\forall E}$, and $m$, where  $\vec{\alpha}=\alpha_1, \ldots, \alpha_k$, for some $k\geq 0$.
\item \label{enum: generation for STAoplus bang } If $\mathcal{D}\triangleleft \Gamma \vdash  \oc M: \sigma $ then $\sigma= \oc \sigma'$, $\Gamma$ is an strictly exponential context,  and  $\mathcal{D}$ is $sp$, followed by some applications of the rule  $m$.
\end{enumerate}
\end{lemma}

 \begin{lemma} \label{lem: exponential context from exponential conclusion in STAoplus}{\ }
\begin{enumerate}
\item  \label{enum: exponential conclusion} If $\mathcal{D}\triangleleft \Gamma \vdash M: \oc \sigma$ then  $\mathcal{D}$ has been obtained from a derivation $\mathcal{D}'$ by applying the rule $sp$, followed by some applications of the rule $m$. Hence, $\Gamma$ is a strictly exponential context and $M= \oc M'$, for some $M'$.
\item \label{enum: x:A is slinear in M}If $\mathcal{D} \triangleleft \Gamma, x: A \vdash M: \tau$ then $x$ is $s$-linear in $M$.
\item \label{enum: exponential context}  If $\mathcal{D}\triangleleft \Gamma, x:\oc \sigma \vdash M: \tau$ then either $x: \oc \sigma$ has been introduced by a $sp$ rule or by a $m$ rule. 
\end{enumerate}
\end{lemma}

Following Gaboardi and Ronchi~\cite{gaboardi2009light},  we prove a \enquote{weighted} formulation of the substitution property. Since we work with two kinds of types, namely the linear types (i.e.~those with form $A$) and the strictly exponential ones (i.e.~those with form $ \sigma$), we  split  the task: first, we consider a substitution theorem for   linear types; then, we generalize the statement to arbitrary types.

 \begin{lemma}[Weighted linear substitution]\label{lem: weighted linear substitution}  Let $r\geq 1$. If $\mathcal{D}_1 \triangleleft \Gamma, x: A \vdash M: \tau$ and $\mathcal{D}_2 \triangleleft \Delta \vdash N: A$, then there exists a derivation $S(\mathcal{D}_1, \mathcal{D}_2)$ such that: 
\begin{itemize}
\item $S(\mathcal{D}_1, \mathcal{D}_2)\triangleleft \Gamma, \Delta \vdash M[N/x]: \tau$, 
\item $\mathrm{w}(S(\mathcal{D}_1, \mathcal{D}_2), r) \leq \mathrm{w}(\mathcal{D}_1,r)+ \mathrm{w}(\mathcal{D}_2, r)$. 
\end{itemize}
\end{lemma}
\begin{proof}
By \textbf{Lemma}~\ref{lem: exponential context from exponential conclusion in STAoplus}.\ref{enum: x:A is slinear in M}, $x$ is $s$-linear in $M$, i.e.~$x$ occurs exactly once  in $M$ and this occurrence is out of the scope of both a  $\oc$-operator and a $\mathtt{d}$-operator.  The statement is proved by induction on $\mathcal{D}_1$. The cases were the last rule is $ax$, $\multimap$I$l$, $\multimap$I$e$, $\multimap$E,  $\with$E, $\forall$I,  $\forall$E, and $m$ are easy.  Now, suppose  $\mathcal{D}_1$ is of the form:
  \begin{prooftree}
 \def\defaultHypSeparation{\hskip .1cm}
 \AxiomC{$\mathcal{D}'$}
 \noLine
\UnaryInfC{$\Gamma, x:A \vdash P: B$}
\AxiomC{$\mathcal{D}''$}
 \noLine
\UnaryInfC{$ x_1:B \vdash Q_1: C_1$}
\AxiomC{$\mathcal{D}''''$}
 \noLine
\UnaryInfC{$ x_2:B \vdash Q_2: C_2$}
\AxiomC{$\mathcal{D}''''$}
 \noLine
\UnaryInfC{$ \vdash V: B$}
\RightLabel{$\with$I}
\QuaternaryInfC{$\Gamma, x:A \vdash \mathtt{copy}^V\, P \mathtt{\ as \ }x_1, x_2 \mathtt{\ in \ }\langle Q_1, Q_2\rangle : C_1 \with C_2$}
\end{prooftree}
so that $\tau=C_1\with C_2$ and $M=\mathtt{copy}^V \, P \mathtt{\ as \ }x_1, x_2 \mathtt{\ in \ }\langle Q_1, Q_2\rangle $. By induction hypothesis, there exists $S(\mathcal{D}', \mathcal{D}_2)\triangleleft   \Gamma , \Delta \vdash P[N/x]: B$ such that $\mathrm{w}(S(\mathcal{D}', \mathcal{D}_2),r)\leq   \mathrm{w}(\mathcal{D}', r)+\mathrm{w}(\mathcal{D}_2, r)$. We define $S(\mathcal{D}_1, \mathcal{D}_2)$ with conclusion:
\begin{equation*}
\Gamma, \Delta \vdash \mathtt{copy}^V\, P[N/x] \mathtt{\ as \ }x_1, x_2 \mathtt{\ in \ }\langle Q_1, Q_2\rangle : C_1 \with C_2
\end{equation*}
 as the derivation obtained by applying $\with$I to $S(\mathcal{D}', \mathcal{D}_2)$, $\mathcal{D}''$, $\mathcal{D}'''$, $\mathcal{D}''''$. Moreover, by using the induction hypothesis, we have:
\allowdisplaybreaks
\begin{align*}
\we{S(\mathcal{D}_1, \mathcal{D}_2)}{r}&=\mathrm{w}(S(\mathcal{D}', \mathcal{D}_2),r)+\we{\mathcal{D}''}{r}+ \we{\mathcal{D}'''}{r}+ \we{\mathcal{D}''''}{r}+2 \\
&\leq \mathrm{w}(\mathcal{D}', r)+\we{\mathcal{D}''}{r}+ \we{\mathcal{D}'''}{r}+ \we{\mathcal{D}''''}{r}+\mathrm{w}(\mathcal{D}_2, r)+2\\
&= \mathrm{w}(\mathcal{D}_1, r)+\we{\mathcal{D}_2}{r}.
\end{align*}
Last, since $A$ is  a linear type,  the last rule of $\mathcal{D}_1$ cannot be $sp$.
\end{proof}

\begin{lemma}[Weighted substitution]\label{lem: weighted substitution} Let  $r \geq \mathrm{rk}(\mathcal{D}_1)$. If  $\mathcal{D}_1 \triangleleft \Gamma, x: \sigma \vdash M: \tau$ and $\mathcal{D}_2 \triangleleft \Delta \vdash N: \sigma$, then there exists a derivation $S(\mathcal{D}_1, \mathcal{D}_2)$ such that:
\begin{itemize}
\item $S(\mathcal{D}_1, \mathcal{D}_2)\triangleleft \Gamma, \Delta \vdash M\lbrace N/x\rbrace : \tau$,
\item $\mathrm{w}(S(\mathcal{D}_1, \mathcal{D}_2), r) \leq \mathrm{w}(\mathcal{D}_1,r)+ \mathrm{w}(\mathcal{D}_2, r)$.
\end{itemize}

\end{lemma}
\begin{proof}
Since $\sigma= \oc^q A$, for some linear type $A$ and some $q \geq 0$, we reason by induction on $q$. If $q=0$ then,  by \textbf{Lemma}~\ref{lem: exponential context from exponential conclusion in STAoplus}.\ref{enum: x:A is slinear in M}, $x$ is $s$-linear in $M$, i.e.~$x$ occurs exactly once  in $M$ and this occurrence is out of the scope of both a  $\oc$-operator and a $\mathtt{d}$-operator. This means that  $M\lbrace N/x\rbrace= M[N/x]$, and we can  apply \textbf{Lemma}~\ref{lem: weighted linear substitution}.  Suppose now that $\sigma= \oc \sigma'$. On the one hand, by \textbf{Lemma}~\ref{lem: exponential context from exponential conclusion in STAoplus}.\ref{enum: exponential conclusion} we have that $\Delta$ is strictly exponential, $N= \oc P$, and $\mathcal{D}_2$ is composed by a subderivation $\mathcal{D}^*_2$ of the form:
\begin{prooftree}
\AxiomC{$\mathcal{D}'_2$}
\noLine
\UnaryInfC{$ \Delta' \vdash P': \sigma'$}
\RightLabel{$sp$}
\UnaryInfC{$  \oc \Delta' \vdash \oc P'[\mathtt{d}(z_1)/y_1, \ldots, \mathtt{d}(z_m)/y_m]: \oc \sigma'$}
\end{prooftree}
with $sp$-rank $h$ and such that $\Delta'=y_1: \sigma_1, \ldots, y_m: \sigma_m$,  followed by a sequence of $t\geq 0$ rules $m$ with $m$-rank, respectively, $k_1, \ldots, k_t$ recovering $\Delta \vdash \oc P: \oc \sigma'$. On the other hand, by applying \textbf{Lemma}~\ref{lem: exponential context from exponential conclusion in STAoplus}.\ref{enum: exponential context},  the assumption $x: \oc \sigma'$ in $\mathcal{D}_1 \triangleleft \Gamma, x: \oc \sigma' \vdash M: \tau$ has been obtained by applying either the rule $sp$ or the rule $m$. We just consider the latter case, the former being similar. W.l.o.g. we can suppose that such an instance of $m$ is the last rule of $\mathcal{D}_1$, since we can always  permute an application of $m$ downward obtaining a derivation of the same judgement. Then, $\mathcal{D}_1$ has the following form:
\begin{prooftree}
\AxiomC{$\mathcal{D}'_1$}
\noLine
\UnaryInfC{$\Gamma, x_1: \sigma',\ldots, x_n: \sigma' \vdash M': \tau $}
\AxiomC{$(n\geq 0)$}
\RightLabel{$m$}
\BinaryInfC{$\Gamma, x:  \oc \sigma'  \vdash M'[\mathtt{d}(x)/x_1, \ldots, \mathtt{d}(x)/x_n]: \tau$}
\end{prooftree}
with $m$-rank $k$ and such that  $M= M'[\mathtt{d}(x)/x_1, \ldots, \mathtt{d}(x)/x_n]$. If  $k=0$ then $S(\mathcal{D}_1, \mathcal{D}_2)$ is $\mathcal{D}'_1$ followed by some applications of the $m$ rule with $m$-rank $0$ in order to recover the context $\Delta$,  which is strictly exponential by \textbf{Lemma}~\ref{lem: exponential context from exponential conclusion in STAoplus}.\ref{enum: exponential conclusion}. In this case, we have  $\mathrm{w}(S(\mathcal{D}_1, \mathcal{D}_2), r)= \mathrm{w}(\mathcal{D}'_1, r) $. Otherwise,    by using the induction hypothesis,  we can build the following derivations:
\allowdisplaybreaks
\begin{align*}
S^1&\triangleq S(\mathcal{D}'_2, \mathcal{D}'_1)\triangleleft \Gamma, \Delta', x_2: \sigma', \ldots, x_n: \sigma' \vdash M'\lbrace P'/x_1\rbrace :\tau\\
S^2&\triangleq S(\mathcal{D}'_2, S(\mathcal{D}'_2, \mathcal{D}_1'))\triangleleft \Gamma, \Delta', \Delta',  x_3: \sigma', \ldots, x_n: \sigma'  \vdash M'\lbrace P'/x_1,P'/x_2\rbrace : \tau\\
&\ldots\\
S^n&\triangleq S(\mathcal{D}'_2, S(\mathcal{D}'_2, \ldots S(\mathcal{D}'_2, \mathcal{D}_1'))) \triangleleft \Gamma, \Delta', \overset{n}{\ldots}, \Delta' \vdash M'\lbrace P'/x_1, \ldots, P'/x_n\rbrace : \tau  
\end{align*}
such that $\mathrm{w}(S^1, r)\leq \mathrm{w}(\mathcal{D}'_2,r)+\mathrm{w}(\mathcal{D}_1',r)$ and, for all $1\leq i <n$,  $\mathrm{w}(S^{i+1}, r)\leq \mathrm{w}(\mathcal{D}'_2, r)+ \mathrm{w}(S^i, r)\leq \mathrm{w}(\mathcal{D}'_1, r)+ (i+1)\cdot \mathrm{w}(\mathcal{D}'_2, r) $.   Then,  $S(\mathcal{D}_1, \mathcal{D}_2)$ can be obtained from $S^n$ by applying a sequence of $h$ applications of the rule $m$ with $m$-rank $k$, and   a sequence of $t$ applications of the rule $m$ with $m$-rank, respectively, $k_1, \ldots, k_t$, in order to get $\Delta$ from $\Delta', \overset{n}{\ldots}, \Delta' $.  This means that $S(\mathcal{D}_1, \mathcal{D}_2)\triangleleft \Gamma, \Delta \vdash   M'\lbrace P/x_1, \ldots, P/x_n\rbrace: \tau$ and, by definition of surface-preserving   substitution:
\allowdisplaybreaks
\begin{align*}
 M'\lbrace P/x_1, \ldots, P/x_n\rbrace&=  (M'[z/x_1, \ldots , z/x_n]) \lbrace   P/z \rbrace\\
&=( (M'[z/x_1, \ldots , z/x_n])[\mathtt{d}(x)/z])) \lbrace \oc  P/x \rbrace &&x \not \in FV(M')\\
 & = (M'[\mathtt{d}(x)/x_1, \ldots, \mathtt{d}(x)/x_n])  \lbrace \oc  P/x \rbrace\\
 &=  M\lbrace N/x \rbrace .
\end{align*}
 By using the induction hypothesis, we finally  have:
\allowdisplaybreaks
\begin{align*}
\mathrm{w}(S(\mathcal{D}_1, \mathcal{D}_2), r)& = \mathrm{w}(S^n, r) +k \cdot h + \sum_{i=1}^t k_i\\
&\leq \mathrm{w}(\mathcal{D}'_1, r)+ k\cdot \mathrm{w}(\mathcal{D}'_2, r) + k\cdot h  +  \sum_{i=1}^t k_i\\
&\leq \mathrm{w}(\mathcal{D}'_1, r)+ r\cdot \mathrm{w}(\mathcal{D}'_2, r) + r\cdot h  +  \sum_{i=1}^t k_i\\
&\leq  \mathrm{w}(\mathcal{D}_1, r)+ (r\cdot( \mathrm{w}(\mathcal{D}'_2, r) + h)+1  +  \sum_{i=1}^t k_i)\\
&= \mathrm{w}(\mathcal{D}_1, r)+ (\mathrm{w}(\mathcal{D}^*_2, r)+  \sum_{i=1}^t k_i)\\
&\leq \mathrm{w}(\mathcal{D}_1, r)+ \mathrm{w}(\mathcal{D}_2, r) .
\end{align*}
This concludes the proof.
\end{proof}

We are now able to prove the weighted version of the Subject reduction property:

\begin{proof}[Proof of {\normalfont\textbf{Theorem}~\ref{thm: weighted subject reduction}}]
The proof is by induction on  the definition of the one-step reduction relation. We have several cases, and we consider the most interesting ones:  
\begin{itemize}
\item  If   $M=(\lambda \oc x.N)\oc P\rightarrow N\lbrace \oc  P/x\rbrace=M_1=M_2$ then, by  applying   \textbf{Lemma}~\ref{lem: generation}.\ref{enum: generation for STA exp absts types} and \textbf{Lemma}~\ref{lem: generation}.\ref{enum: generation for STA app types}, $\mathcal{D}$ contains a derivation  $\mathcal{D}^*$  of the form:
\begin{prooftree}
\AxiomC{$\mathcal{D}'$}
\noLine
\UnaryInfC{$\Gamma' , x: \tau\vdash  N':   A$}
\RightLabel{$\multimap$I$e$}
\UnaryInfC{$\Gamma' \vdash \lambda \oc x.N': \tau \multimap A$}
\noLine
\AxiomC{$\mathcal{D}''$}
\noLine
\UnaryInfC{$\Gamma'' \vdash \oc P': \tau$}
\RightLabel{$\multimap$E}
\BinaryInfC{$\Gamma', \Gamma'' \vdash  (\lambda \oc x.N')\oc P': A$}
\end{prooftree}
possibly followed by a sequence of applications of the rules $\forall$I, $\forall$E, and  $m$. Let  $t\geq 0$ be the number of applications of the rule $m$, and let $k_1,\ldots, k_t$ be their respective $m$-rank. By applying \textbf{Lemma}~\ref{lem: weighted substitution}, there exists a derivation $S(\mathcal{D}', \mathcal{D}'')$ such that $S(\mathcal{D}', \mathcal{D}'')\triangleleft \Gamma', \Gamma'' \vdash N'\lbrace \oc P'/x\rbrace :A$. We define $\mathcal{D}_1=\mathcal{D}_2$ as the derivation obtained by applying to $S(\mathcal{D}', \mathcal{D}'')$ a sequence of applications of the rules $\forall$I, $\forall$E, and  $m$ in order to obtain $\Gamma \vdash N\lbrace \oc P/x \rbrace :\sigma$  as a concluding judgement.  By \textbf{Lemma}~\ref{lem: weighted substitution},  we have:
 \allowdisplaybreaks
\begin{align*}
\mathrm{w}(\mathcal{D}_1, r)&=\mathrm{w}(S(\mathcal{D}', \mathcal{D}''), r)+ \sum_{j=1}^t k_j\leq  \mathrm{w}(\mathcal{D}', r)+ \mathrm{w}(\mathcal{D}'', r)+ \sum_{j=1}^t k_j \\
  &<  \mathrm{w}(\mathcal{D}', r)+ \mathrm{w}(\mathcal{D}'', r)+ \sum_{j=1}^t k_j +2=  \mathrm{w}(\mathcal{D},r) .
\end{align*}
\item If $M= \mathtt{proj}\langle M_1, M_2\rangle \rightarrow M_1, M_2$ then, by applying     \textbf{Lemma}~\ref{lem: generation}.\ref{enum: generation for STA copy} and \textbf{Lemma}~\ref{lem: generation}.\ref{enum: generation for STA pi},  $\sigma=\forall \vec{\alpha}.( B'\langle D_1/\beta_1, \ldots, D_n/\beta_n \rangle)$, where  $\vec{\alpha}=\alpha_1, \ldots, \alpha_k$, for some $k\geq 0$. Moreover,   $\mathcal{D}$ is a derivation $\mathcal{D}^*$ of the form:
\begin{prooftree}
\AxiomC{$\mathcal{D}'$}
\noLine
\UnaryInfC{$\vdash M_1: B$}	
\AxiomC{$\mathcal{D}''$}
\noLine
\UnaryInfC{$\vdash M_2: B$}
\RightLabel{$\with$I}
\BinaryInfC{$ \vdash \langle M_1, M_2\rangle :B\with B$}
\RightLabel{$\wedge$E}
\UnaryInfC{$ \vdash \mathtt{proj}\langle M_1, M_2\rangle : B$}
\end{prooftree}
followed  by  a sequence of applications of the rules $\forall$I, $\forall$E, and $m$. Then, we define $\mathcal{D}_1$ (resp. $\mathcal{D}_2$) as the derivation  $\mathcal{D}'$  (resp. $\mathcal{D}''$) followed by the same sequence of rules $\forall$I, $\forall$E, and  $m$, the latter being of $m$-rank $0$ and introducing  the context $\Gamma$. By definition of weight, we have:
$\mathrm{w}(\mathcal{D}_1, r)=\mathrm{w}(\mathcal{D}', r)<\w{\mathcal{D},r}$, and similarly for $\mathcal{D}_2$.
\item If $M=\mathtt{copy}^{U}\, V \mathtt{\ as \ }x_1, x_2 \mathtt{\ in \ }\langle Q_1, Q_2\rangle  \rightarrow \langle Q_1[V/x_1], Q_2[V/x_2]\rangle =M_1=M_2 $ then, by \textbf{Lemma}~\ref{lem: generation}.\ref{enum: generation for STA copy},  $\sigma= B_1\with B_2$ and  $\mathcal{D}$ is a derivation $\mathcal{D}^*$ of the form:
  \begin{prooftree}
 \def\defaultHypSeparation{\hskip .1cm}
 \AxiomC{$\mathcal{D}'$}
 \noLine
\UnaryInfC{$ \Gamma' \vdash V: A$}
 \AxiomC{$\mathcal{D}''$}
 \noLine
\UnaryInfC{$x_1:A \vdash Q_1: B_1$}
 \AxiomC{$\mathcal{D}'''$}
 \noLine
\UnaryInfC{$ x_2:A \vdash Q_2: B_2$}
 \AxiomC{$\mathcal{D}''''$}
 \noLine
\UnaryInfC{$\vdash U: A$}
\RightLabel{$\with$I}
\QuaternaryInfC{$\Gamma' \vdash \mathtt{copy}^{U}\, V \mathtt{\ as \ }x_1, x_2 \mathtt{\ in \ }\langle Q_1, Q_2\rangle : B_1 \with B_2$}
\end{prooftree}
followed by  a sequence of applications of the rule $m$. Since  $\Gamma'$ is $\forall \oc $-lazy by definition, it is  $\oc$-free, and hence all types in $\Gamma'$ are linear. Then, since   $V$ is closed,  \textbf{Lemma}~\ref{lem: exponential context from exponential conclusion in STAoplus}.\ref{enum: x:A is slinear in M}   implies   $\Gamma'= \emptyset$. Therefore,   the applications of the rule $m$ below $\mathcal{D}^*$ are all of $m$-rank $0$, so that $\w{\mathcal{D},r}= \w{\mathcal{D}^*,r}$. By   applying \textbf{Lemma}~\ref{lem: weighted linear substitution} twice,  there exist two derivations $S(\mathcal{D}', \mathcal{D}'')\triangleleft \vdash Q_1[V/x_1]:B_1$ and $S(\mathcal{D}', \mathcal{D}''')\triangleleft \vdash Q_2[V/x_2]: B_2$ such that $\w{S(\mathcal{D}', \mathcal{D}''), r}\leq \w{\mathcal{D}',r}+ \w{\mathcal{D}'',r} $ and $\w{S(\mathcal{D}', \mathcal{D}'''), r}\leq \w{\mathcal{D}',r}+ \w{\mathcal{D}''',r}$. We define $\mathcal{D}_1=\mathcal{D}_2$ as the following derivation:
\begin{prooftree}
\AxiomC{$S(\mathcal{D}', \mathcal{D}'')$}
\noLine
\UnaryInfC{$ \vdash Q_1[V/x_1]:B_1$}
\AxiomC{$S(\mathcal{D}', \mathcal{D}''')$}
\noLine
\UnaryInfC{$\vdash Q_2[V/x_2]: B_2$}
\RightLabel{$\with$I}
\BinaryInfC{$\vdash \langle Q_1[V/x_1],Q_2[V/x_2] \rangle : B_1\with B_2$}
\RightLabel{$m$}
\doubleLine
\UnaryInfC{$\Gamma \vdash \langle Q_1[V/x_1],Q_2[V/x_2] \rangle : B_1\with B_2$}
\end{prooftree}
By \textsf{Remark}~\ref{rem: w.l.o.g. largest values} we can safely assume that $U$ has largest size among the  values with type $A$. Moreover, $\mathcal{D}'$ and $\mathcal{D}'''$ have no application of the rules $sp$ and $m$  so that,  by \textbf{Lemma}~\ref{lem: properties of weight}.\ref{enum: properties of weight 3}, $\w{\mathcal{D}', r}= \vert V \vert \leq \vert U \vert=\w{\mathcal{D}'''', r}$.   Therefore:
\allowdisplaybreaks
\begin{align*}
\w{\mathcal{D}_1, r}&=\w{S(\mathcal{D}', \mathcal{D}''), r}+\w{S(\mathcal{D}', \mathcal{D}'''), r}+1\\
& \leq 2\cdot \w{\mathcal{D}', r} +  \w{\mathcal{D}'', r}+\w{\mathcal{D}''',r} +1\\
&\leq   \w{\mathcal{D}', r} + \w{\mathcal{D}'', r}+  \w{\mathcal{D}''', r}+\w{\mathcal{D}'''',r}+1 \\
&<   \w{\mathcal{D}', r} + \w{\mathcal{D}'', r}+  \w{\mathcal{D}''', r}+\w{\mathcal{D}'''',r}+2 \\
&= \w{\mathcal{D}^*, r}   = \w{\mathcal{D}, r} .
\end{align*}
\end{itemize}
This concludes the proof.
\end{proof}

\begin{proof}[Proof of {\normalfont\textbf{Lemma}~\ref{lem: all derivations of Rightarrow have same size}}]
The proof is by induction on $\s{\pi'}+\s{\pi''}$. If the last rule of $\pi'$ is $s1$ then $M$ is a surface normal form, and the last rule of $\pi''$ must be $s1$. In this case,  $\vert \pi' \vert =0= \vert \pi'' \vert$. If the last rule of $\pi'$ is $s2$, then $M$ is not a surface normal form, so that the last rule of $\pi''$ is $s2$. Hence, $\pi'$ and $\pi''$ have the following forms:
\begin{mathpar}
\inferrule*[Right=$s2$]{M \rightarrow M'_1, M'_2\\ \pi'_1: M'_1 \Rightarrow \mathscr{D}'_1\\ \pi'_2: M'_2 \Rightarrow \mathscr{D}'_2}{\pi': M \Rightarrow \mathscr{D}}\and
\inferrule*[Right=$s2$]{M \rightarrow M''_1, M''_2\\ \pi''_1: M''_1 \Rightarrow \mathscr{D}''_1\\ \pi''_2: M''_2 \Rightarrow \mathscr{D}''_2}{\pi'': M \Rightarrow \mathscr{D}}
\end{mathpar}
We have several possibilities depending on $M'_1, M'_2, M''_1, M''_2$. We  just consider the  case  where they are all distinct. By applying \textbf{Lemma}~\ref{lem: one-reduction step confluence 3} there exist $N_1, N_2, N_3, N_4$ such that $M'_1 \rightarrow N_1, N_2$, $M'_2 \rightarrow N_3, N_4$ and $\exists i \in \lbrace 1,2\rbrace$ such that $M''_i \rightarrow N_1, N_3$ and $M''_{3-i}\rightarrow N_2, N_4$. Let us suppose $i=1$. By \textbf{Theorem}~\ref{thm: weighted subject reduction} $N_1$, $N_2$, $N_3$ and $N_4$ are all typable in \STAP. Moreover, since each typable term can be associated with exactly one surface distribution by \textbf{Theorem}~\ref{thm: confluence for lamb} and \textbf{Theorem}~\ref{thm: weighted subject reduction}, for all $1 \leq j\leq 4$, we have   $\rho_j: N_j \Rightarrow \mathscr{E}_j$, for some $\rho_j$ and $\mathscr{E}_j$. Then, we can construct the following derivations:
\begin{mathpar}
\inferrule*[Right=$s2$]{ M'_1 \rightarrow N_1, N_2 \\ \ \rho_1:N_1  \Rightarrow \mathscr{E}_1\\   \rho_2:N_2 \Rightarrow \mathscr{E}_2  }{\rho'_1: M'_1\Rightarrow \mathscr{D}'_1 }\and 
\inferrule*[Right=$s2$]{ M'_2 \rightarrow N_3, N_4\\ \ \rho_3:N_3  \Rightarrow\mathscr{E}_3 \\ \rho_4  :N_4 \Rightarrow  \mathscr{E}_4 }{\rho'_2: M'_2 \Rightarrow \mathscr{D}'_2} \and 
\inferrule*[Right=$s2$]{ M''_1 \rightarrow N_1, N_3 \\ \ \rho_1:N_1  \Rightarrow \mathscr{E}_1\\ \rho_3  : N_3 \Rightarrow\mathscr{E}_3   }{ \rho''_1: M''_1 \Rightarrow\mathscr{D}''_1 }
\and 
\inferrule*[Right=$s2$]{ M''_2 \rightarrow N_2, N_4 \\ \rho_2 :N_2  \Rightarrow \mathscr{E}_2\\  \rho_4 : N_4 \Rightarrow  \mathscr{E}_4 }{ \rho''_2: M''_2 \Rightarrow \mathscr{D}''_2}
\end{mathpar}
By applying the induction hypothesis we have:
\allowdisplaybreaks
\begin{align*}
\s{\pi'}&= \max (\s{\pi'_1}, \s{\pi'_2})+1\\
&= \max (\s{\rho'_1}, \s{\rho'_2})+1 \\
&=\max(\max (\s{\rho_1}, \s{\rho_2})+1, \max(\s{\rho_3}, \s{\rho_4})+1)+1\\
&=\max(\max (\s{\rho_1}, \s{\rho_3})+1, \max(\s{\rho_2}, \s{\rho_4})+1)+1\\
&= \max (\s{\rho''_1}, \s{\rho''_2})+1\\
&=\max (\s{\pi''_1}, \s{\pi''_2})+1 = \s{\pi''}.
\end{align*}
The remaining cases are similar.
\end{proof}

\section{Proofs of Section~\ref{section:Probabilistic Polytime Completeness}}

In this section we give a detailed proof of the Probabilitic Polytime  Completeness Theorem for \STAP (\textbf{Theorem}~\ref{thm: completeness complexity PTM}).  The basic scheme of the proof is taken from Gaboardi and Ronchi Della Rocca~\cite{gaboardi2009light}, and consists in encoding {\PTM}s configurations,  transitions between configurations, the initialization of a \PTM, and its output extraction. By putting everything together, we are able to represent in \STAP a  \pPTM.  Before giving  the complete encoding, we shall first show how to define in \STAP  natural numbers and polynomials.

\subsection{Numerals and polynomial completeness}\label{chap 2 sec 4 subsec 1}
Gaboardi and Ronchi Della Rocca stressed in~\cite{gaboardi2009light}   that the    presence of the multiplexor, i.e.~rule $m$,   makes  the  encoding of a Turing Machine    \enquote{non-uniform}  in  $\mathsf{STA}$.  If we consider for example the standard type for natural numbers $\mathbf{N}\triangleq \forall \alpha. \oc (\alpha \multimap \alpha)\multimap \alpha \multimap \alpha$, 
 a term  $\mathtt{succ}$ implementing  the usual successor function  with  type  $\mathbf{N}\multimap \mathbf{N}$ is unknown.   This is why the usual data types  are  represented in \STAP   by  indexed families of types.   

\begin{definition}[Indexed  numerals] For all $i \geq 1$, the \emph{indexed type} $\mathbf{N}_i$ and the \emph{indexed  numerals} $\underline{n}_i$ of type $\mathbf{N}_i$ are defined as follows:
\begin{equation*}
\begin{split}
\mathbf{N}_i &\triangleq \forall \alpha. \oc^i(\alpha \multimap \alpha)\multimap \alpha \multimap \alpha \\
  \underline{n}_i &\triangleq \lambda \oc f.\lambda x. (\mathtt{d}^i(f) \overset{n}{\ldots} (\mathtt{d}^i(f) x)\ldots) \qquad n \in \mathbb{N} 
\end{split}
\end{equation*}
when $i=1$, we shall  write $\mathbf{N}$ (resp.~$\underline{n}$) in place of $\mathbf{N}_i$ (resp.~$\underline{n}_i$). 
\end{definition}

\begin{definition} \label{defn: succ add mult} Let $i, j \geq 1$.  The \emph{indexed successor} $\mathtt{succ}_i$ of type $ \mathbf{N}_{i}\multimap \mathbf{N}_{i+1}$, the \emph{indexed  addition} of type $\mathbf{N}_i \multimap \mathbf{N}_j \multimap \mathbf{N}_{{\max(i, j)+1}}$, and  the \emph{indexed multiplication} of type \emph{$\mathbf{N}_i\multimap \oc^i \mathbf{N}_j \multimap \mathbf{N}_{i+j} $} are definable in \normalfont{\STAP} as follows:
\begin{itemize}
\item $\mathtt{succ}_i\triangleq  \lambda n. \lambda \oc f. \lambda x. \mathtt{d}^{i+1}(f)(n \, (\oc^{i} \mathtt{d}^{i+1}(f)) \, x)$;
\item $\mathtt{add}_{i, j}\triangleq \lambda n. \lambda m. \lambda \oc f. \lambda x. n \, (\oc ^{i}\mathtt{d}^{\max(i, j)+1}(f))(m \, (\oc ^{j}\mathtt{d}^{(\max(i, j)+1}(f)) \, x)$;
\item $\mathtt{mult}_{i, j}\triangleq \lambda n. \lambda m. \lambda \oc f. n \, \oc ^i(m \, (\oc ^j \mathtt{d}^{i+j}(f))) $.
\end{itemize}
\end{definition}

Successor,  addition, and  multiplication in \textbf{Definition}~\ref{defn: succ add mult} can  be composed  to obtain all  polynomials.

\begin{theorem}[Representing polynomial functions~\cite{gaboardi2009light}]
\label{thm: polynomial completeness} 
Let $p: \mathbb{N}\longrightarrow \mathbb{N}$ be a polynomial in the variable $\mathrm{x}$ and $\deg(p)$ be its degree. 
There is $\underline{p}$ such that:
\begin{equation*}
x: \oc^{\deg(p)}\mathbf{N} \vdash \underline{p}: \mathbf{N}_{2\deg(p)+1} \enspace .
\end{equation*} 
\end{theorem}

Booleans and indexed strings of booleans are defined, respectively, in~\eqref{eqn: booleans}  and~\eqref{eqn: indexed strings}. The function  associating with each  string of booleans its length is defined for all $i \geq 1$ as follows:
\begin{equation}\label{eqn: len}
\mathtt{len}_i\triangleq \lambda s. \lambda f. s\, \oc ^i(\lambda x. \lambda y. \mathtt{let}\ \mathtt{E}_{\mathbf{B}}\, x \mathtt{\ be \ }\mathbf I \mathtt{ \ in \ }fy)  
\end{equation}
with type $\mathbf{S}_i \multimap \mathbf{N}_i $, where $\mathtt{E}_{\mathbf{B}}$ is as follows:
\begin{equation}\label{eqn: eraser of B}
\mathtt{E}_{\mathbf{B}} \triangleq
\lambda z. \mathtt{let\ }z\mathbf I\mathbf I \mathtt{\ be \ } x\otimes y \mathtt{ \ in \ }(\mathtt{let \ }y \mathtt{ \ be \ } \mathbf I \mathtt{ \ in \ }x) :\mathbf{B}\multimap \mathbf{1}
\end{equation}

 \subsection{Encoding the $\pPTM$}\label{chap 2 sec 4 subsec 2}
In this subsection we show how to encode a $\pPTM$ in \STAP and how to simulate its computation by means of the relation $\Rightarrow$ in  \textbf{Definition}~\ref{defn: small-step probabilsitic operational semantics for surface reduction}.  One of the key steps toward completeness is  to prove that every $\PTM$  transition function is definable in \STAP, and its encoding is in~\eqref{equation: transition function PTM}.

A configuration can be represented  by a tuple   divided up  in three parts: the first one represents  the left hand-side of the tape with respect to the head; the second one  represents the right part of the tape starting  with the cell  scanned by the head; finally, the third part represents the state of the machine. W.l.o.g.,~we shall assume that  the left part of the tape is represented in  reversed order, that the alphabet is composed by the two symbols $0$ and $1$,  and that the final states are divided into accepting and rejecting.
 
\begin{definition}[Indexed configuration]\label{defn: configurations}
For all $i, k \geq 1$, we define the  \emph{indexed type} $\mathbf{PTM}^k_i$  and the  \emph{indexed configuration}  $\mathtt{config}_i$  of type $\mathbf{PTM}^k_i$ as  follows:
\begin{equation*}
\begin{split}
\mathbf{PTM}^k_i& \triangleq \forall \alpha. \oc^i (\mathbf{B}\multimap \alpha \multimap \alpha)\multimap ((\alpha \multimap \alpha) ^2 \otimes \mathbf{B}^k)\\
\mathtt{config}_{i} &\triangleq \lambda \oc c. \,  (\mathtt{d}^i(c)\, \underline{b^l_0}   \circ \cdots \circ \mathtt{d}^i(c)\, \underline{b^l_{n}})\otimes ( \mathtt{d}^i(c)\, \underline{b^r_0 }\circ \cdots \circ \mathtt{d}^i(c)\, \underline{b^r_{m}})\otimes  \underline{Q}   .
\end{split}
\end{equation*} 
where $M \circ N\triangleq \lambda z. M(Nz)$, $\underline{Q} \triangleq  \underline{q_1}\otimes  \ldots \otimes \underline{q_k}$, and $b_0^l, \ldots, b_n^l, b_0^r, \ldots, b^r_m, q_1, \ldots, q_k$ are in $\lbrace 0, 1 \rbrace$, for  $n, m \in \mathbb{N}$.
\end{definition}

In the above definition, the terms:
\begin{equation*}
\mathtt{d}^i(c)\, \underline{b^l_0} \circ \cdots \circ \mathtt{d}^i(c)\, \underline{b^l_{n_l}} \qquad  \mathtt{d}^i(c)\, \underline{b^r_0} \circ \cdots \circ \mathtt{d}^i(c)\, \underline{b^r_{n_r}} \qquad  \underline{Q} \triangleq  \underline{q_1}\otimes  \ldots \otimes \underline{q_k}
\end{equation*}
represent, respectively,  the left and the right part of the tape, where $\mathtt{d}^i(c)\, \underline{b^r_0}$ is the scanned symbol, and   the current state $Q=(q_1, \ldots, q_k)$.

Following Mairson and Terui~\cite{mairson2003computational}, in order to define the  \PTM transition from a configuration to another     we consider two distinct phases. In the first one, the \PTM configuration is decomposed to extract the first symbol of each part of the tape.  In the second phase, depending on the transitions function, these  symbols  are combined to reconstruct the tape after the transition step. Thus, we require an intermediate type, denoted $\mathbf{ID}^k_i$, and  defined for all $i,k \geq 1$ as follows:
\begin{equation*}
 \forall \alpha. \oc^i(\mathbf{B}\multimap \alpha \multimap \alpha) \multimap  ((\alpha \multimap \alpha)^2 \otimes (\mathbf{B}\multimap \alpha \multimap \alpha) \otimes \mathbf{B}\otimes (\mathbf{B}\multimap \alpha \multimap \alpha) \otimes \mathbf{B}  \otimes \mathbf{B}^k)
\end{equation*}
and the decomposition phase is defined by the term $\mathtt{decom}_{i}$  of type $\mathbf{PTM}^k_i \multimap \mathbf{ID}^k_i $ below:
\begin{equation}\label{eqn: decomp}
\begin{aligned}
\mathtt{decom}_{i} \triangleq \ &\lambda m. \lambda \oc c.  \mathtt{let}\ m \, \oc^i(F[\mathtt{d}^i(c)]) \mathtt{\ be \ } l\otimes r\otimes  q \mathtt{\ in \ }   \\
&(\mathtt{let\ }  l\,  (\mathbf I\otimes  (\lambda x. \mathtt{let}\ \mathtt{E}_{\mathbf{B}}\, x \mathtt{\ be \ } \mathbf I \mathtt{\ in\ } \mathbf I)\otimes  \underline{0} ) \mathtt{\ be \ }s_{l}\otimes  c_{l}\otimes  b^l_0 \mathtt{\ in \ }  \\
&  (\mathtt{let}\ r\, ( \mathbf I\otimes  ( \lambda x. \mathtt{let}\ \mathtt{E}_{\mathbf{B}}\, x \mathtt{\ be \ } \mathbf  I \mathtt{\ in\ } \mathbf  I) \otimes  \underline{0} ) \mathtt{\ be \ }s_r\otimes  c_r\otimes  b^r_0 \mathtt{\ in \ } \\
&  s_l\otimes   s_r\otimes   c_l\otimes  b^l_0\otimes  c_r\otimes  b^r_0\otimes  q    ))
\end{aligned} 
\end{equation}
where $F[x]\triangleq \lambda b. \lambda z. \mathtt{let}\ z \mathtt{\ be \ } g\otimes  h\otimes  i \mathtt{ \ in \ } ( hi \circ g) \otimes  x\otimes  b $ and $\mathtt{E}_{\mathbf{B}}$ is as in~\eqref{eqn: eraser of B}.

The behaviour of $\mathtt{decom}_i$ is to decompose a configuration in such a way as to extract the symbols of the tape which determine, together with the current state, the structure of the next configuration:
\allowdisplaybreaks
\begin{align*}
&\mathtt{decom}_{i} (\lambda \oc c.\, ( \mathtt{d}^i(c)\, \underline{b^l_0} \circ C[b^l_1, \ldots b^l_n]) \otimes ( \mathtt{d}^i(c)\, \underline{b^r_0 }\circ C[b^r_1, \ldots b^r_m]) \otimes \underline{Q}  )\\
&\Rightarrow 	\lambda \oc c. \, C[b^l_1, \ldots b^l_n] \otimes C[b^r_1, \ldots b^r_m] \otimes  \mathtt{d}^i(c)\otimes  \underline{b^l_0}\otimes \mathtt{d}^i(c)\otimes \underline{b^r_0}\otimes   \underline{Q}      
\end{align*}
where $C[b^l_1, \ldots b^l_n]\triangleq  \mathtt{d}^i(c)\, \underline{b^l_1}   \circ  \cdots \circ \mathtt{d}^i(c)\, \underline{b^l_n}$ and $C[b^r_1, \ldots b^r_m]\triangleq  \mathtt{d}^i(c)\, \underline{b^r_1 }\circ \cdots \circ \mathtt{d}^i(c)\, \underline{b^r_m}$.

Analogously,  the composition phase is defined  by the  term $\mathtt{com}_{i}$ of type $\mathbf{ID}^k_i \multimap \mathbf{PTM}^k_i $ below:
\begin{equation}
\begin{aligned}
\mathtt{com}_{i}\triangleq \ &\lambda s. \lambda \oc c. \mathtt{let}\ s\, \oc^i(\mathtt{d}^i(c)) \mathtt{\ be\ } l\otimes  r\otimes   c_l\otimes  b_l\otimes  c_r\otimes  b_r\otimes   q \mathtt{ \ in \ } \mathtt{let}\  \underline{\delta_{\mathcal{P}}}\,   (b_r\otimes  q)  \\
& \mathtt{\ be \ }q'\otimes b'\otimes m \mathtt{ \ in \ }  (\mathtt{if}\ m \mathtt{ \ then \ } M_1 \mathtt{\ else \ } M_2 )\, b'\, q' ( l\otimes  r\otimes c_l\otimes  b_l\otimes  c_r ) 
\end{aligned}
\end{equation}
where  $\underline{\delta_{\mathcal{P}}}$ is the encoding of the transition function $\delta_{\mathcal{P}}$ of the \PTM   $\mathcal{P}$ as in~\eqref{equation: transition function PTM}, and: 
\begin{equation*}
\begin{aligned}
& \mathtt{if\ } x \mathtt{\ then \ } M_1 \mathtt{\ else \ }M_2\triangleq  \pi_1(x \, M_1\, M_2)\\
&\pi_1\triangleq  \lambda z. \mathtt{let\ }z \mathtt{\ be \ }x\otimes y \mathtt{\ in\ }(\mathtt{let \ }\mathtt{E}_{\mathbf{B}} \, y \mathtt{\ be \ }\mathbf I \mathtt{\ in \ }x) \\
&M_1\triangleq \lambda b'. \lambda q'. \lambda p. \mathtt{let}\ p \mathtt{\ be \ }l\otimes  r\otimes  c_l\otimes  b_l\otimes  c_r\mathtt{ \ in \ }   (c_r\, b' \circ c_l\, b_l \circ l)\otimes  r\otimes q' \\
&M_2 \triangleq \lambda b'. \lambda q'. \lambda p. \mathtt{let}\ p \mathtt{\ be \ }l\otimes  r\otimes   c_l\otimes  b_l\otimes  c_r \mathtt{ \ in \ }   l\otimes  (c_l\, b_l \circ c_r\, b' \circ r) \otimes  q'   .
\end{aligned}
\end{equation*}

Then, the behaviour of $\mathtt{com}_{i}$, depending on $\delta_\mathcal{P}$  and on the current state, is to combine the symbols we put aside in order to return  a distribution of the next configurations. For example, if  the deterministic transition functions $\delta_0$ and $\delta_1$ defining $\delta_\mathcal{P}$ are such that $\delta_0\, ((b^r_0, Q))=(Q', b', \mathrm{right})$ and $\delta_1\, ((b^r_0, Q))=(Q'', b'', \mathrm{left})$, then:
\begin{equation*}
\begin{split}
&\mathtt{com}_{i}\,  (	 C[b^l_1, \ldots b^l_n] \otimes C[b^r_1, \ldots b^r_m] \otimes  \mathtt{d}^i(c)\otimes  \underline{b^l_0} \otimes   \mathtt{d}^i(c)\otimes   \underline{b^r_0}\otimes    \underline{Q}     )\\
& \Rightarrow \frac{1}{2}\cdot  \lambda \oc  c.\,    (\mathtt{d}^i(c)\, \underline{b'}\circ  \mathtt{d}^i(c)\, \underline{b^l_0} \circ  C[b^l_1, \ldots b^l_n]) \otimes C[b^r_1, \ldots b^r_m] \otimes \underline{Q'}   \\
&\hspace{5cm}+ \\
&\phantom{\Rightarrow\ {}} \frac{1}{2}\cdot \lambda \oc c.\,  C[b^l_1, \ldots b^l_n]  \otimes (  \mathtt{d}^i(c)\, \underline{b^l_0 }\circ   \mathtt{d}^i(c)\, \underline{b'' }\circ  C[b^r_1, \ldots b^r_m]) \otimes  \underline{Q'}   .
\end{split}
\end{equation*}
where $C[b^l_1, \ldots b^l_n]\triangleq  \mathtt{d}^i(c)\, \underline{b^l_1}   \circ  \cdots \circ \mathtt{d}^i(c)\, \underline{b^l_n}$ and $C[b^r_1, \ldots b^r_m]\triangleq  \mathtt{d}^i(c)\, \underline{b^r_1 }\circ \cdots \circ \mathtt{d}^i(c)\, \underline{b^r_m}$.

By combining the above terms we obtain an entire \PTM  transition step. 
\begin{definition}[Indexed transition step] \label{defn: tr} Let $i, k\geq 1$. The \emph{indexed transition step} is defined by $\mathtt{tr}_{i}\triangleq \mathtt{com}_{i}\circ \mathtt{decom}_{i}$, with type $\mathbf{PTM}^k_i \multimap \mathbf{PTM}^k_i$ in \normalfont{\STAP}.
\end{definition}

The initial configuration of a \PTM  is a  configuration in the initial state  $Q_0=(q_1, \ldots, q_k)$  with the head at the beginning of a tape  filled by   $0$'s. Then,  we need  a term that, taking  a   numeral  $\underline{n}_i$ as input, gives  the encoding of the initial configuration with tape of length $n$ as output. 
\begin{definition}[Indexed initial configuration]\label{defn: init} For all $i,k \geq 1$, the \emph{indexed initial configuration} $\mathtt{init}_{i}$ of type $\mathbf{N}_i \multimap \mathbf{PTM}^k_i$ is defined as follows:
\begin{equation*}
\mathtt{init}_{i} \triangleq \lambda n.  \lambda \oc c. \,  (\lambda z. z) \otimes ( \lambda z.n \, \oc^i (\mathtt{d}^i(c)\, \underline{0})z) \otimes \underline{Q_0}  .
\end{equation*}
\end{definition}

The \PTM needs now to be initialized with the given input string,  by writing it on its tape. The term representing the initialization requires the term $ \mathtt{decom}_{i}$ in~\eqref{eqn: decomp}.

\begin{definition}[Indexed initialization]\label{defn: in} For all $i,k\geq 1$, the \emph{indexed initialization}  is defined by $\mathtt{in}_{i} \triangleq \lambda s. \lambda m. s\,  \oc(\lambda b.Tb \circ \mathtt{decom}_{i})\, m $ of type $\mathbf{S}\multimap \mathbf{PTM}^k_i \multimap \mathbf{PTM}^k_i$, where: 
\begin{equation*}
\begin{aligned}
T  \triangleq &\   \lambda b. \lambda  m. \lambda \oc c. \mathtt{let}\ m\ (\oc^i \mathtt{d}^i(c)) \mathtt{\ be\  }l \otimes  r\otimes  c_l\otimes  b_l\otimes  c_r\otimes  b_r\otimes q  \mathtt{ \ in\  } \\ 
&\ \mathtt{let}\ \mathtt{E}_{\mathbf{B}} \, b_r \mathtt{\ be \ } \mathbf I \mathtt{ \ in \ } Rbq \, ( l\otimes  r\otimes  c_l\otimes  b_l\otimes  c_r )\\
R \triangleq& \  \lambda b'. \lambda q'. \lambda p. \mathtt{let}\ p \mathtt{\ be \  } l\otimes  r\otimes   c_l\otimes  b_l\otimes  c_r\mathtt{ \ in  \ }\,  (c_r \, b' \circ c_l \, b_l \circ l)\otimes  r\otimes  q'   
\end{aligned}
\end{equation*}
where $\mathtt{E}_{\mathbf{B}}$ is as in~\eqref{eqn: eraser of B}.
\end{definition}

Last, we need to extract the output string from the final configuration. 

\begin{definition}[Indexed extraction] \label{defn: extrfun} For all $i, k \geq 1$,  we define the \emph{indexed extraction} $\mathtt{ext}_i^{\mathbf{S}}$ of type $\mathbf{PTM}^k_i \multimap \mathbf{S}_i$ as the following term:
\begin{equation*}
\begin{aligned}
\mathtt{ext}_i^{\mathbf{S}}\triangleq & \ \lambda m. \lambda \oc c. \mathtt{let}\ m\, \oc^i(\mathtt{d}^i(c)) \mathtt{\ be \ }l\otimes  r\otimes  q \mathtt{ \ in \ }(\mathtt{let \ }\mathtt{E}_{\mathbf{B}^k}\, q \mathtt{\ be \ }\mathbf{I} \mathtt{ \ in \ }l \circ r) .
\end{aligned}
\end{equation*}
where $\mathtt{E}_{\mathbf{B}^k}$ has type  $\mathbf{B}^k \multimap \mathbf{1}$, and can be constructed from~\eqref{eqn: eraser of B}.
\end{definition}

By putting everything together, we are now able to encode  a  \pPTM in \STAP:

\begin{proof}[Proof of {\normalfont\textbf{Theorem}~\ref{thm: completeness complexity PTM}}]
Let $\mathcal{P}$ be a \PTM  running in polynomial time $p: \mathbb{N}\longrightarrow \mathbb{N}$ and in polynomial space $q: \mathbb{N}\longrightarrow\mathbb{N}$, with $\deg(p)=d_1$ and $\deg(q)=d_2$. We set  $[p]= 2d_1+1$ and $[q]=2d_2+1$. By \textbf{Theorem}~\ref{thm: polynomial completeness} and \textbf{Lemma}~\ref{lem: weighted substitution}  we have that the following judgements are derivable in \STAP:
\begin{equation}\label{eqn: polynomioal in the length}
\begin{split}
&s_p: \oc^{d_1} \mathbf{S} \vdash P: \mathbf{N}_{[p]}\\
&s_q: \oc^{d_2} \mathbf{S} \vdash Q: \mathbf{N}_{[q]}  
\end{split}
\end{equation}
where $P\triangleq \underline{p}\, \lbrace \oc^{d_1} (\mathtt{len}_1\, \mathtt{d}^{d_1}(s_p))/x\rbrace$, $Q\triangleq \underline{q}\, \{\oc^{d_2}(\mathtt{len}_1\, \mathtt{d}^{d_2}(s_q))/x\}$, and  $\mathtt{len}_1$ is defined in~\eqref{eqn: len}. Again, by repeatedly applying  \textbf{Lemma}~\ref{lem: weighted substitution} we can compose the terms in \textbf{Definitions}~\ref{defn: tr},~\ref{defn: init},~\ref{defn: in}, and~\ref{defn: extrfun} to obtain a derivation in \STAP of the following judgement:
\begin{equation}\label{eqn: intermediate encoding PTM}
s': \mathbf{S}, p: \mathbf{N}_{[p]}, q: \mathbf{N}_{[q]} \vdash \mathtt{ext}_{[q]}^{\mathbf{S}}(p\,    ( \oc^{[p]} \mathtt{tr}_{[q]}) (\mathtt{in}_{[q]} \, s'\, (\mathtt{init}_{[q]} \, q) )): \mathbf{S}_{[q]} .
\end{equation}
 By two further applications  of   \textbf{Lemma}~\ref{lem: weighted substitution},  we can compose~\eqref{eqn: polynomioal in the length} and~~\eqref{eqn: intermediate encoding PTM} to obtain the following: 
\begin{equation*}
s': \mathbf{S},s_p: \oc^{d_1} \mathbf{S},s_q:\oc^{d_2} \mathbf{S}   \vdash \mathtt{ext}_{[q]}^{\mathbf{S}}(P\,    ( \oc^{[p]} \mathtt{tr}_{[q]}) (\mathtt{in}_{[q]} \, s'\, (\mathtt{init}_{[q]} \,Q) )): \mathbf{S}_{2d_2+1}.
\end{equation*}
By repeatedly applying rule $m$, and by applying rule $\multimap$I$l$, we obtain the  term:
\[  \vdash_{\normalfont{\STAP}} \underline{\mathcal{P}}: \oc^{\max( d_1,  d_2, 1)+1}\mathbf{S}\multimap \mathbf{S}_{2d_2+1}. \]
 One can  check that both point~\textsf{\textbf{\ref{enum: completeness 1}}} and point~\textsf{\textbf{\ref{enum: completeness 2}}}  hold.
\end{proof}


\end{document}